\documentclass[12pt]{article}

\pdfoutput=1

\usepackage{amsthm,enumerate}
\usepackage{mathtools}
\usepackage[normalem]{ulem}
\usepackage[T1]{fontenc}
\usepackage[utf8]{inputenc}
\usepackage[thinc]{esdiff}
\usepackage{color}
\usepackage{fancyhdr}
\usepackage{calc}
\usepackage{url}
\usepackage{xcolor}
\usepackage{xspace}

\usepackage{enumitem}

\usepackage{array}

\usepackage[all,cmtip]{xy} %can be removed
\usepackage{tikz-cd}

\usepackage[top=80pt,bottom=85pt,left=85pt,right=85pt]{geometry}
\usepackage{amsmath,amssymb,graphicx,float,color,tikz,cite,savesym,wasysym,subfigure,wrapfig}
\usepackage{caption}
\usepackage[debug,pageanchor=false]{hyperref}
\definecolor{link}{rgb}{1,0.45,0.05}
\hypersetup{colorlinks=true,linkcolor=link,citecolor=link,urlcolor=link,linktocpage}

\newcommand{\ii}{\mathrm{i}}
\newcommand{\dd}{\mathrm{d}}
\newcommand{\ee}{\mathrm{e}}

\newcommand{\RCD}{\mathsf{RCD}}

\DeclareMathOperator*{\sgn}{sgn}

\newcommand{\mm}{\mathfrak m}
\newcommand{\di}{\mathsf{d}}

\theoremstyle{plain}
\newtheorem{lemma}{Lemma}[section]
\newtheorem{theorem}[lemma]{Theorem}

\newtheorem{proposition}[lemma]{Proposition}

\newtheorem*{theorem*}{Theorem}
\newtheorem*{maintheorem*}{Main Theorem}

\theoremstyle{definition}

\newtheorem{definition}[lemma]{Definition}
\newtheorem*{definition*}{Definition}
\newtheorem*{remark*}{Remark}
\newtheorem{remark}[lemma]{Remark}

\makeatletter
\@addtoreset{equation}{section}
\makeatother

\begin{document}

	\begin{titlepage}

	\begin{center}

	\vskip .5in %.3in 
	\noindent

%	{\Large \bf{Weyl law, gravitational potential\\ and quantum ergodicity}}
	{\Large \bf{Can you hear the Planck mass?}}

	\bigskip\medskip
	 G. Bruno De Luca,$^1$  Nicol\`o De Ponti,$^2$
	Andrea Mondino,$^3$	Alessandro Tomasiello$^{2,4}$\\

	\bigskip\medskip
	{\small 
$^1$ Stanford Institute for Theoretical Physics, Stanford University,\\
382 Via Pueblo Mall, Stanford, CA 94305, United States
\\	
	\vspace{.3cm}%{.1cm}
	$^2$ 
	Dipartimento di Matematica e Applicazioni, Universit\`a degli Studi di Milano--Bicocca, \\ Via Cozzi 55, 20126 Milano, Italy
\\	
	\vspace{.3cm}%{.1cm}
	$^3$ Mathematical Institute, University of Oxford, Andrew-Wiles Building,\\ Woodstock Road, Oxford, OX2 6GG, UK
\\
	\vspace{.3cm}
$^4$ INFN, sezione di Milano--Bicocca
		}

   \vskip .5cm %.3cm
	{\small \tt gbdeluca@stanford.edu, nicolo.deponti@unimib.it,\\ andrea.mondino@maths.ox.ac.uk, alessandro.tomasiello@unimib.it}
	\vskip .9cm %.6cm
	     	{\bf Abstract }
	\vskip .1in
	\end{center}
	
	\noindent
	For the Laplacian of an $n$-Riemannian manifold $X$, the Weyl law states that the $k$-th eigenvalue is asymptotically proportional to $(k/V)^{2/n}$, where $V$ is the volume of $X$.
	We show that this result can be derived via physical considerations by demanding that the gravitational potential for a compactification on $X$ behaves in the expected $(4+n)$-dimensional way at short distances. In simple product compactifications, when particle motion on $X$ is ergodic, for large $k$ the eigenfunctions oscillate around a constant, and the argument is relatively straightforward. The Weyl law thus allows to reconstruct the four-dimensional Planck mass from the asymptotics of the masses of the spin 2 Kaluza--Klein modes.
		For warped compactifications, a puzzle appears: the Weyl law still depends on the ordinary volume $V$, while the Planck mass famously depends on a weighted volume obtained as an integral of the warping function. We resolve this tension by arguing that in the ergodic case the eigenfunctions oscillate now around a power of the warping function rather than around a constant, a property that we call \emph{weighted quantum ergodicity}. This has implications for the problem of gravity localization, which we discuss.
We show that for spaces with D$p$-brane singularities the spectrum is discrete only for $p =6,7,8$, and for these cases we rigorously prove the Weyl law by applying modern techniques from RCD theory.
			\noindent

	\vfill
	\eject

	\end{titlepage}
    
\tableofcontents
\section{Introduction} % (fold)
\label{sec:intro}
% section intro (end)

In gravitational models with $n$ extra dimensions, a four-dimensional observer sees a tower of Kaluza--Klein (KK) massive particles. The latter give rise to Yukawa-type potentials, that decay exponentially; when the distance  $r$ between particles is very large, they can be ignored, and the gravitational potential $U$ is proportional to the familiar $1/(m_4^2 r)$, with $m_4$ the four-dimensional Planck mass.\footnote{We focus on $d=4$ uncompactified dimensions, but a similar logic applies to $d\neq 4$.} On the contrary, when $r$ is much smaller than the size of the extra dimensions, one expects that the KK particles conspire to reproduce the behavior of gravity in $D=(4+n)$ spacetime dimensions, namely $U\sim m_D^{2-D}/r^{D-3}\propto m_D^{2-D}/r^{1+n}$, with $m_D$ the $D$-dimensional Planck mass.

We will see that such an expectation implies the \emph{Weyl law}, a well-known property of the Laplace operator $\Delta$: on an $n$-dimensional Riemannian manifold $X$, the $k$-th eigenvalue (counted with multiplicity) behaves at large $k$ as 
\begin{equation}\label{eq:weyl}
	\lambda_k \sim a^2 \left(\frac{k}{V(X)}\right)^{2/n} \, ,\qquad a = \frac{2\pi}{\omega_n^{1/n}} \,,
\end{equation}
where $\omega_n := \pi^{n/2}/ \Gamma(1+n/2)$ is the volume of the $n$-dimensional Euclidean ball. The exponent $2/n$ is precisely the right value so that $U\propto 1/r^{1+n}$. The term $V(X)$  in the coefficient in (\ref{eq:weyl}) is the usual Riemannian volume; it helps convert the four-dimensional Newton constant into the $D$-dimensional one, since for direct-product spacetimes $m_4^2=m_D^{D-2} V(X)$. Thanks to this property, knowing the asymptotic behavior of the spin 2 Kaluza--Klein masses allows to determine the four-dimensional Planck mass in terms of the higher-dimensional Planck mass.

This gravitational argument can be seen as a variant of a more classical argument using the heat equation. In its traditional form, the latter requires at some point to take an average over $X$; this step is needed because of the  appearance of the eigenfunctions $\psi_k$ of the Laplace operator. This can also be done for our case.
However, it is more natural to consider the behavior of particles that are localized at some point of $X$. 
As we will see, if $X$ has the \emph{ergodic} property, which is expected to be generic, one can straightforwardly derive the Weyl law without needing to integrate over $X$. Classically, this means roughly speaking that the geodesic motion at large time explores uniformly all of $X$ (or more precisely all of phase space). This also implies \emph{quantum} ergodicity \cite{schnirelman,colindeverdiere,zelditch-surfaces}: roughly speaking this means that, for large $k$, the eigenfunction $\psi_k^2$ oscillates around a constant value. As a consequence, for any measurable open set $B\subset X$ the integral $\int_B |\psi_k|^2 \to V(B)/V(X)$. (More precisely, this is true up to a subset of measure zero of the indexes $k$.) 

Everything looks natural so far. Our work was actually motivated by a puzzle that appears in the general case of warped compactifications, namely those where the total $D$-dimensional metric is not a direct product, but reads as $\dd s^2= \ee^{2A}(\dd s^2_{M_4} + \dd s^2_{X})$; the warping $A$ is a function of $X$, and is only defined up to the shift $A \mapsto A - A_0$, for a constant $A_0$, since such a transformation can be reabsorbed by rescaling $\dd s^2_{M_4}$. The relevant KK square masses are now eigenvalues of the weighted Laplacian \cite{bachas-estes,csaki-erlich-hollowood-shirman}:
$$\Delta_f\psi := -\ee^{-f} \nabla^m (\ee^{f} \nabla_m\psi),\quad \text{with $f:=(D-2)A$}.$$
The relation between the Newton constants is now $G_4 \propto G_D/V_f(X)$, where
\begin{equation}\label{eq:wV}
	V_f (X):= \int_{X}\dd^n y\sqrt{g_n}\ee^{f} 
\end{equation}
can be thought of as a \emph{weighted volume} ($y$ being the internal coordinates). However, known mathematical results show that the eigenvalues of $\Delta_f$ still obey the Weyl law (\ref{eq:weyl}) with the ordinary, unweighted volume $V(X)$. As we will review in Sec.~\ref{sec:weyl-rev}, this has been established even in situations where $X$ has certain types of singularities. 

Naively, it would then seem that our gravitational argument for the Weyl law would break down for warped compactification: dressing with a warping the physical quantities would predict to substitute the volumes with warped volumes, thereby predicting the appearance of the Planck mass in the law, as in the unwarped case.

The resolution turns out to hinge on the role of the wavefunctions $\psi_k$. We saw earlier that quantum ergodicity predicts them to oscillate around a constant value; but that statement has been proven for the ordinary curved-space Laplace operator. For the weighted Laplacian $\Delta_f$, which is the one relevant to warped compactifications, this property does not hold. However, analyzing the dependence of the gravitational potential on the internal position of the test particles suggests a replacement: at large $k$, the $\psi_k^2$ actually tend to oscillate instead around a non-constant function proportional to $\ee^{-f}$, as depicted in Fig.~\ref{fig:WQE-1d}. It is well-known that the natural measure for the inner products and norms includes a factor $\ee^{f}$. When we consider the squared norm on a measurable $B$, we obtain
\begin{equation}\label{eq:WQE-intro}
	\frac{\int_B \sqrt{g}\ee^{f}\psi_k^2}
	{\int_X \sqrt{g}\ee^{f}\psi_k^2}\quad 
	 \underset{k\to \infty}{\longrightarrow} \quad \frac{V(B)}{V(X)} \,.
\end{equation}
Notice that the ordinary volume $V(X)$ still appears. It is natural to call this the \emph{weighted quantum ergodicity} (WQE) property. 
\begin{figure}[h]
	\centering
	  \includegraphics[width=9cm]{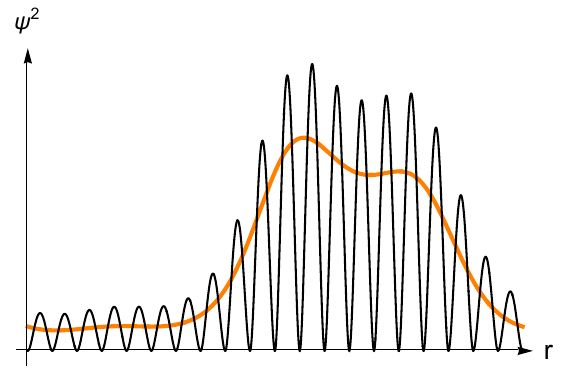}
	\caption{Numerical check of weighted quantum ergodicity for the example in App.~\ref{sub:ex-erg}, where $M_1=\mathbb{S}^1$ and $f=\sin(x)+\cos^3x$. In orange, the function $\ee^{-f}/V$ predicted by (\ref{eq:WQE-intro}); in black, $\psi^2_k/V_f$ for $k=20$.}
	 \label{fig:WQE-1d}
  \end{figure}
While we have not proven the WQE, we have tested it in some simple models, as in the example in Fig.~\ref{fig:WQE-1d}, which we analyzed in detail in App.~\ref{sub:ex-erg}. Moreover, it is simple to see that the gravitational potential now exactly reproduces the physics expectation: the factor of $V(X)$ in the Weyl law is provided by (\ref{eq:WQE-intro}). This explains why the Weyl law is not in contradiction with the physics of gravity compactifications. Thus, in the general case of warped compactifications, the knowledge of the asymptotics of the KK spectrum \emph{does not allow to reconstruct the Planck mass}, as summarized in Table \ref{tab:recap}.

\begin{table}[!ht]
	\centering
	\begin{tabular}{l|c|c|c|c|}
	\cline{2-5} &Spin 2 operator
	& Planck Mass & Weyl Law & Ergodicity \\ \hline
	\multicolumn{1}{|l|}{Unwarped} & $\Delta$& $m_4^2  = m_D^{D-2} V$& $m_k^2  \sim V^{-2/n} k^{2/n}$& $\frac{\int_B \sqrt{g}\psi_k^2}
	{\int_M \sqrt{g}\psi_k^2} 
	 \to   \frac{V(B)}{V(X)}$\\ \hline
	\multicolumn{1}{|l|}{Warped} &$\Delta - \nabla f \cdot \nabla$ & $m_4^2  = m_D^{D-2} V_f$ & $m_k^2  \sim V^{-2/n} k^{2/n}$ & 
	$\frac{\int_B \sqrt{g}\ee^{f}\psi_k^2}
	{\int_M \sqrt{g}\ee^{f}\psi_k^2} 
	 \to   \frac{V(B)}{V(X)} $ \\ \hline
	\end{tabular}
	\caption{A comparison of different quantities between warped and unwarped compactifications. Notice that in the warped case the coefficient in the Weyl law does \emph{not} correspond to the Planck mass.}
	\label{tab:recap}
	\end{table}

We begin in Sec.~\ref{sec:grav} by setting the stage and deriving (\ref{eq:grav-weyl}), a relation involving a certain sum over eigenvalues and eigenfunctions. We give both a more physical gravitational argument, and one that is more precise but partially abandons physical intuition. In Sec.~\ref{sec:grav-weyl} we apply (\ref{eq:grav-weyl}) to obtain the Weyl law. Again we give two versions: a less physical argument gets the result quickly but involves smearing over the internal space. 
The more physical unsmeared argument requires some knowledge about properties of eigenfunctions, which leads us in Sec.~\ref{sub:wqe} to the idea of weighted quantum ergodicity. As an application, we point out in Sec.~\ref{sub:loc} that WQE might also help establishing gravity localization in some models, since it gives a natural way to suppress the graviton eigenfunctions in some regions.

While these physical derivations allow to prove the Weyl law under certain hypotheses, they are not rigorous mathematical proofs for the general case, in particular in the presence of singularities. Thus, in Sec.~\ref{sec:weyl-rev}, we review known mathematical results concerning the Weyl law in the RCD setting, and apply them to rigorously prove the Weyl law for general compactifications with D$p$-brane singularities, for $p = 6,7, 8$. For $p<6$ we formalize in Prop.~\ref{prop:Dbranediscrete} our earlier results that show the spectrum is instead continuous. While an appropriate rigorous generalization of the Weyl law might be given also for these case of continuum spectrum, we do not attempt this here.

We conclude in Sec.~\ref{sec:conc}. App.~\ref{app:ex} contains some explicit examples on the validity of weighted quantum ergodicity, and App.~\ref{app:heat} reviews an older mathematical argument involving the heat equation, for the Weyl law in the unwarped case. Finally, App.~\ref{app:green} contains a technical lemma regarding Green's functions for the weighted Laplacian.

\section{Gravitational potential} % (fold)
\label{sec:grav}

After some brief preliminaries on the general setting, we will discuss the gravitational potential of a compactification. This will result in a physical expectation on the spectrum, which we will use in the next section to constrain the asymptotic behavior of the eigenvalues.

\subsection{Preliminaries}
\label{sub:prel}

%{\color{red}NDP: Penso sia da ritoccare con un paio di frasette questa sezione, che così sembra un po' monca.}

The Weyl law was originally established more than a century ago for the Laplace operator on bounded domains in two dimensions. Subsequent work generalized this to higher dimensions, to Riemannian manifolds, and gave information about the subleading behavior. See for example \cite{ivrii-review} for a historical review.

Here we also need the generalization regarding the weighted Laplacian, on a space $X$ with weight function $f$:
\begin{equation}\label{eq:w-Lap}
	\Delta_f \psi := -\ee^{-f}\nabla_m (\ee^f\nabla^m \psi) = -\frac1{\ee^f\sqrt{g}} \partial_m ( \ee^f\sqrt{g} g^{mn} \partial_n \psi)\,.
\end{equation}
Its eigenvalues and eigenfunctions are defined by $\Delta_f \psi_k = \lambda_k \psi_k$. We will focus on the case where the spectrum is discrete, and $k$ runs over the integers.\footnote{While some of our arguments below can probably be extended to cases where the spectrum also has a continuous part by using the spectral measure, this case appears to be much less studied mathematically in our setting.} This is the case when $X$ is smooth and compact, but as we will see below this holds also in the presence of many types of singularities.

As we saw in the introduction, the Weyl law remains (\ref{eq:weyl}) also for the weighted case $f \neq 0$. A perhaps slightly clearer mathematical formulation is
\begin{equation}\label{eq:weyl-count}
	N(\lambda):=\#\{k\in \mathbb{N} : \lambda_k < \lambda\} \sim  \frac{V(X)}{a^n} \lambda^{n/2} \qquad\textrm{as } \lambda\to+\infty \,,
\end{equation}
where the eigenvalues should be counted with multiplicity.

In the next sections, we will derive a relation among eigenvalues and eigenfunctions that will serve as the basis for obtaining the Weyl law in the rest of the work.

\subsection{Compactification} % (fold)
\label{sub:comp}

We consider a $D$-dimensional theory that includes the usual Einstein--Hilbert action $S_D= m_D^{D-2}\int \dd^D x\sqrt{-g_D}R_D$. To compute the potential between two mass sources $M_{1,2}$ we will solve the equations of motion in the presence of $M_1$, and compute the gravitational potential felt by a probe $M_2$. For this, we will use the source actions $S_{M_{1,2}} \equiv - M_{1,2} \int_{\Sigma}\sqrt{-g_D\rvert_{\Sigma}}$, where $\Sigma$ denotes the world-volume of the source. 
\newline

\noindent We will compare two situations:
\begin{enumerate}[label=\roman*)]
	\item Higher-dimensional space-time: The space-time has the topology of $\mathbb{R}^{1,D-1}$ with the Minkowski metric.
	\item Four-dimensional vacuum compactification: The space-time has the topology of $\mathcal{M}_4\times X$. The product is generally warped, with a warping depending on the $(D-4)$-dimensional internal space $X$.
\end{enumerate}
Denoting with $\rho$ the proper distance between the two particles, we expect the gravitational potentials to agree in the limit $\rho \to 0$, since in the compactified theory we expect to reproduce higher-dimensional gravity at distances shorter than the compactification scale.

We start with i), and we work in a regime where $M_{1,2}\ll m_D$, or in other words where the test masses have a Schwarzschild radius much smaller than their Compton wavelength. We then have $g_{MN}\sim \eta_{MN}+m_D^{2-D} h_{MN}$, $h_{MN}\ll m_D^{D-2}$. 
In the TT gauge, where $\nabla^M h_{MN}=0=h^M_M$, at quadratic level in $h$, the action describing the backreaction of $M_1$ on the background reads 
\begin{equation}
	\delta^2S = \frac{m_D^{2-D}}{2}\int \sqrt{-\eta_D} h^{M N}\nabla^2 h_{M N} +m_D^{2-D} \frac{M_1}{2} \int_\Sigma (-\eta_{00})^{-1/2}h^{00} \;.
\end{equation}
Since the source is static, we can specialize to time-independent perturbations, resulting in the following equation of motion for $h_{00}$:
\begin{equation}
	\Delta_{D-1} h_{00} = \frac{M_1}{2} \delta_{D-1}\,,
\end{equation} 
where $\delta_{D-1}$ is a delta function localizing $M_1$. Notice that since we are in flat space, $-\nabla^2$ reduces to the scalar Laplacian.
Thus, $h_{00}$ is a Green's function of the scalar Laplacian in $\mathbb{R}^{D-1}$. We will re-use this fact in the more precise derivation of the Weyl law in Sec.~\ref{sub:grav-sc}. 

We are left with the task of determining the potential felt by a probe $M_2$. Expanding the probe action at the same order we get $U(x_0) = M_2 \left(1- \frac{m_D^{2-D}}{2} h_{00}(x_0)\right)$. 
Using radial coordinates in $\mathbb{R}^{D-1}$ centered on $M_1$ we have 
\begin{equation}
	h_{00} = - \frac{M_1}{2} \frac{c}{\rho^{D-3}} \, ,\qquad
	c:= \frac1{(D-3)V(\mathbb{S}^{D-2})}\,,
\end{equation}
resulting in the gravitational potential
\begin{equation}
	U_D(\rho) = M_2 \left(1- M_1 \frac{m_D^{2-D}}{4} \frac{c}{\rho^{D-3}}\right)\,.
\end{equation}

We can now turn our attention to the case ii) of a  \emph{warped compactification}, namely, a spacetime with a line element $\dd s^2 = \ee^{2A}\left(\dd s^2_4 + \dd s^2_{X}\right)$, with $A$ a function on $X$. We take $X$ to be compact and $\dd s^2_4$ to be the Minkowski$_4$ line element; the other maximally symmetric spacetimes, AdS$_4$ and dS$_4$, can also be accommodated, with suitable adjustments. (As already mentioned in the introduction, changing the number of macroscopic uncompactified dimensions is also possible.) $D$-dimensional gravity can be approximated at large distances by a four-dimensional action $S_4+ S_\mathrm{KK}$. Here $S_4=m_4^2\int \dd^4 x\sqrt{-g_4}R_4$, with
\begin{equation}\label{eq:GG}
	m_4^2= m_D^{D-2}V_f(X)\,;
\end{equation}
recall that $V_f$ is the weighted volume (\ref{eq:wV}). $S_\mathrm{KK}$ represents the contribution of infinitely many matter fields; in particular, there are infinitely many spin-two fields.\footnote{A detailed discussion of $S_\mathrm{KK}$ for the pure gravity case with zero warping is given in \cite{hinterbichler-levin-zukowski}. An earlier review with a broader outlook is \cite{duff-nilsson-pope}.} While the full general expression of this action is not known when the warping is non-trivial, its quadratic expansion around vacua is universal. Following \cite{bachas-estes,csaki-erlich-hollowood-shirman} and reducing the source action, we find
\begin{equation}
	\begin{split}
	2 m_4^2\delta^2 S =&   \sum_{k} \int_{\mathcal{M}_4} \sqrt{-\bar{\eta}}\left(\bar{h}_k^{\mu \nu} \bar{\nabla}^2 \bar{h}_{\mu \nu}^k + m_k^2 \bar{h}_k^{\mu \nu}  \bar{h}_{\mu \nu}^k\right) \\&+\sum_k M_1 e^{A(y_0)} \psi_k(y_0) \int_{\mathcal{M}_4} (-\bar{\eta}_{00})^{-1/2}\bar{h}_{00}^k \delta_3\;.
\end{split}
\end{equation}
To write this action, we have considered transverse-traceless perturbations of the \emph{unwarped} four-dimensional metric, that is the $D$-dimensional metric is perturbed as
\begin{equation}
	\dd s^2 = \ee^{2A}\left( (\bar{\eta}_{\mu \nu}+m_4^2 \bar{h}_{\mu \nu}(x,y)) \dd x^{\mu} \dd x^{\nu}+ \dd s^2_{X}(y)\right)	\;
\end{equation}
where $\bar{\eta}_{\mu \nu}$ is the four-dimensional Minkowski metric, possibly multiplied by a constant that we will fix later. In addition, $m_k^2$ and $\psi_k$ are respectively eigenvalues and eigenfunctions of the weighted Laplacian (\ref{eq:w-Lap}) relative to the weight function \cite{bachas-estes,csaki-erlich-hollowood-shirman}
\begin{equation}\label{eq:fA}
	f:= (D-2) A= (n+2)A\,.
\end{equation}
We have normalized the eigenfunctions so that 
\begin{equation}\label{eq:psi-norm}
	\int_{X} \sqrt{g}\ee^f \psi_k \psi_l= V_f(X) \delta_{kl}\,,
\end{equation}
where $V_f(X)$ is the weighted volume (\ref{eq:wV}).
Since $\Delta_f$ in (\ref{eq:w-Lap}) is self-adjoint, its eigenfunctions are an $L^2$-basis for functions in the internal space, and we have expanded $\bar{h}_{\mu \nu}(x,y):= \sum_k \bar{h}^k_{\mu \nu}(x)\psi_k(y)$. Finally, we have also assumed that the particle $M_1$ sits at the internal position $y = y_0$, and called $\delta_3$ the delta-function that localizes it in non-compact space.
From this action, we obtain that the $\bar{h}_{00}^k$ satisfy the equation of motion
\begin{equation}\label{eq:green4dph}
	(\bar{\Delta}_3 - m_k^2) \bar{h}^k_{0 0} = \frac{M_1}{2} \ee^{A (y_0)} \psi_k
	(y_0) \frac{\delta_3 (x)}{\sqrt{\bar{\eta}_3} } (- \bar{\eta}_{00})\,,
\end{equation}
where, as above, we have assumed static perturbations. 
We are now interested in the potential energy  $U_{4d}$ between $M_1$ and a probe mass $M_2$ which sits at the same internal point $y_0$ and is separated by a proper distance $\rho$ from $M_1$ in the uncompactified directions.  Expanding the probe action, we find $	U_{4d}(x_0) =  M_2 e^{A (y_0)} \sqrt{- \bar{\eta}_{00}} \left( 1
- \frac{m_4^{- 2}}{2}  \frac{\bar{h}_{00} (x_0, y_0)}{- \bar{\eta}_{00}} 
\right)
$. Solving \eqref{eq:green4dph}, we get for $m_k \neq 0$ the Yukawa expression
\begin{equation}
	\bar{h}_{00}^k = - 
	\frac{M_1}{8\pi}
	\frac{\ee^{- m_k r}}{r} \ee^{A (y_0)} \psi_k (y_0) (-
\bar{\eta}_{00})\,,
\end{equation}
where $r$ is the four-dimensional radial distance from $M_1$, computed with respect to the metric $\bar{\eta}$.
Before plugging it in the gravitational potential, recall that we need to express the gravitational potential in terms of the proper distance $\rho$. At leading order, this is related to $r$ by $\rho = \ee^{A(y_0)} r$. In addition, before comparing the two potentials at short distances, we have to make sure that we are computing the energy with respect to the same reference. That is, we also have to require that $U_{4d} \sim U_D$ for $\rho \to \infty$. This fixes $\bar{\eta}_{00} = - \ee^{-2A(y_0)}$. 
All in all, we find
\begin{equation}
	U_{4d}(\rho)%\xrightarrow[\rho \sim 0]{} 
	\sim U_D(\rho)\quad \implies\quad %c_2 
	m_4^{- 2} e^{(D - 2) A (y_0)} \sum_k \ee^{- m_k r} \psi_k^2 (y_0)
	 \sim 4\pi c\, m_D^{2 - D} \frac{1}{r^n}\,,
	\quad \text{as } r\to 0 \,.
\end{equation} 

Recalling also (\ref{eq:GG}) to compare the Planck masses, this gives the prediction\footnote{We have also rewritten the numerical factors using the identity $(n+1)!V(\mathbb{S}^{n+2})V(\mathbb{S}^{n-1})= 2 (2\pi)^{n+1}n$, which in turn uses the duplication formula for $\Gamma(z)$. Recall that $\omega_n= V(B_n)$, the volume of the $n$-dimensional Euclidean ball $B_n$.}
\begin{equation}\label{eq:grav-weyl}
    \sum_{k=0}^\infty \ee^{-m_k r} \psi_k^2(y_0) e^{f(y_0)} \quad \sim \quad  \frac{n!\omega_n V_f(X)}{(2\pi r)^n} \qquad \text{as }r\to 0\,.
\end{equation}

In Sec.~\ref{sec:grav-weyl} we will use this physical expectation to derive the Weyl law in various ways.

% subsection comp (end)

\subsection{Scalar version} % (fold)
\label{sub:grav-sc}

The essence of the gravitational argument leading to \eqref{eq:grav-weyl} involves the comparison of the Green's function of (weighted) Laplacians in different dimensions, close to their poles. In this section, we dispense with the gravitational intuition and derive \eqref{eq:grav-weyl} by directly comparing these Laplacians in a more rigorous way.

While for simplicity we will use the language of Riemannian geometry and state our results for smooth manifolds, the results in this section apply to more general $n$-dimensional spaces $X$, which can be smooth Riemannian manifolds of finite diameter, asymptotically D$p$-branes (for $p\geqslant 6$) \cite[Def.~6.1]{deluca-deponti-mondino-t-entropy}, or smooth spaces with O-plane singularities, provided the spectrum is discrete and that the Green's functions are centered at smooth points.

To start, consider a product $p$-dimensional Riemannian manifold $M_p \equiv \mathbb{R}^3\times X$, where $X$ has dimension $n = p-3$, with product metric
\begin{equation}
\dd s^2_p(z) = \delta_{ij}\dd x^i \dd x^j	+\dd s^2_n(y)\,	.
\end{equation}
We use coordinates $z=(x,y)$ to denote points on the different factors. We can think of this space as a spatial slice of the `unwarped' $D$-dimensional space-time of Sec.~\ref{sub:comp}.

On this space we will compare the behavior of $G_{0,z_0}$, the Green's function of the standard Laplacian $\Delta_0$ centered at $z_0$, with the behavior of $G_{f,z_0}$, the Green's function of the weighted Laplacian $\Delta_f \equiv \Delta_0 -\nabla f \cdot \nabla$ centered at the same point.
Since on a vacuum the warping can only depend on the internal coordinates, we take $f = f(y)$, so that
\begin{equation}\label{eq:fact}
	\Delta_f = -\partial_x^2+\Delta_f^{(n)},
\end{equation}
where $\Delta_f^{(n)}$ is a weighted Laplacian on $X$ with weight $f(y)$. In particular, we assume that it is self-adjoint and that its eigenfunctions $\psi_k$, defined by $\Delta_f^{(n)} \psi_k(y) \equiv m_k^2 \psi_k(y)$, provide a countable $L^2$-basis on which to expand $L^2$-functions on $X$. In particular, assuming the Green's function to be square integrable (a sufficient condition for this is that the internal space $X$ has finite diameter), we can write
\begin{equation}\label{eq:Gf}
G_{f, (x_0, y_0)}(x,y)\equiv \sum_k c_{f,k}(x)\psi_k(y)	
\end{equation}
where $z_0 = (x_0,y_0)$ and the limit in the series is understood in the $L^2$-topology. Using the factorization property \eqref{eq:fact} we can write $c_{f,k}(x)$ explicitly in terms of the Yukawa potential as follows.
Take a test function $\xi$ proportional to an arbitrary eigenfunction $\psi_j$,  $\xi(x,y) \equiv \xi_j(x)\psi_j(y)$ (no sum over $j$). By definition of Green's function, we have
\begin{align}% \label{eq:gf1}
	\xi_j(x_0)\psi_j(y_0)  &=  \int_{\mathbb{R}^3} \dd^3x \int_{X} \sqrt{g_n} \dd^ny \,\ee^f\;\left(\sum_k c_{f,k}(x) \psi_k(y) \Delta_f \left(\xi_j(x)\psi_j(y)\right) \right) \nonumber\\
	&=  \int_{\mathbb{R}^3} \dd^3x \int_{X} \sqrt{g_n} \dd^ny \,\ee^f \;\left(\sum_k  c_{f,k}(x) \psi_k(y)\psi_j(y) \left(-\partial_x^2 + m_j^2\right)\xi_j(x)  \right)\nonumber\\
	&=  V_f(X) \int_{\mathbb{R}^3} \dd^3x  \;\left(  c_{f,j}(x)  \left(-\partial_x^2 + m_j^2\right)\xi_j(x)  \right)	\label{eq:gf1}
\end{align}
where we normalized the eigenfunctions $\psi_k$ as $\int_{X} \sqrt{g_n} \ee^f \psi_k \psi_j = V_f(X)\delta_{ij}$,  with $V_f(X)$ the weighted volume of $X$. In the second step, exchanging the summation and the integral is justified since the series \eqref{eq:Gf} converges in $L^2$ when the diameter is finite.
We can recognize that by definition $c_{f,j}$ is proportional to the Green's function of  the operator $-\nabla^2_x+m_j^2$ in $\mathbb{R}^3$, centered at $x_0$. This has the form of the Yukawa potential, and explicitly we have 
\begin{equation}\label{eq:cj}
	c_{f,j}(x) = - \frac{\psi_j(y_0)}{V_f(X)}\frac{\ee^{- m_j r}}{4 \pi r} \qquad \qquad \text{with} \;r \equiv |x-x_0|\;.
\end{equation}
Plugging \eqref{eq:cj} in \eqref{eq:Gf} we get an expression for $G_{f, (x_0, y_0)}(x,y)$ in terms of the spectral data of $\Delta_f^{(n)}$. Before using it to derive \eqref{eq:grav-weyl}, we also need the following Lemmas.

\begin{lemma}\label{lemma:ef0}
	Call $G_{f, z_0}$ the Green's function of the operator $\Delta_f$ centered at
$z_0$, and $G_{0, z_0}$ the Green's function of the standard Laplacian
$\Delta_0$ on the same $p$-dimensional Riemannian manifold $M_p$. If $f$ is smooth at $z = z_0$, then 
\begin{equation}\label{eq:C1}
  \lim_{z \rightarrow z_0} \frac{G_{f, z_0} (z)}{G_{0, z_0} (z)} = e^{- f
  (z_0)} .
\end{equation}
\end{lemma}
\noindent An intuitive understanding of this result can be obtained by noticing that the dominant local behavior of the Green's function close to the pole is not affected by the weight, since the weight only adds to the Laplacian a term with a single derivative. For a smooth weight, this term is subdominant at short distances to the two-derivative terms already present in the Laplacian. Thus, only the pointwise value of $f$ is important, but it does not alter the power-law behavior of the Green's function.
For a similar reason, on a smooth Riemannian manifold the local behavior of the Green's function of the Laplacian close to the source is not sensitive to the details of the metric, thus behaving as in $\mathbb{R}^n$. More precisely, we have the following
\begin{lemma}\label{lemma:Rn}
	Call $G_{0, z_0}$  the Green's function of the Laplacian on a smooth $p$-dimensional Riemannian manifold $M_p$, centered at $z = z_0$, then 
	\begin{equation}\label{eq:lemmaRn}
		 \lim_{z\to z_0}|z-z_0|^{p-2}G_{0, z_0}(z)  = - \frac{1}{4} \pi^{- \frac{p}{2}} \Gamma \left( \frac{p}{2} - 1 \right).
	\end{equation} 
	Equivalently, $G_{0, z_0}(z)$ approaches the Green's function in $\mathbb{R}^p$ as $z\to z_0$.
\end{lemma}
\noindent See App.~\ref{app:green} for a proof of these Lemmas.

Applying the two lemmas to $\Delta_f$ and $\Delta_0$, and combining them with \eqref{eq:cj} and \eqref{eq:Gf} we finally get
\begin{equation}\label{eq:scal-weyl}
\begin{split}
	\ee^{-f(y_0)} &= \lim_{r\to 0 }\lim_{y\to y_0} \frac{\sum_k - \frac{\psi_k(y_0)}{V_f(X)}\frac{\ee^{- m_k r}}{4 \pi r} \psi_k (y)}{- \frac{1}{4} \pi^{- \frac{3+n}{2}} \Gamma \left( \frac{n+1}{2} \right)|z-z_0|^{-n-1}} \\
	&=\frac{\pi^{ \frac{1+n}{2}}}{V_f(X)\Gamma \left( \frac{n+1}{2}  \right)} \lim_{r\to 0} r^{n}\sum_k  \psi^2_k(y_0)\ee^{- m_k r}\,,
\end{split}	
\end{equation}
where in the first line we have split the limit for $z\to z_0$  as a limit for $y\to y_0$ followed by a limit for $r=|x-x_0|\to 0$. Notice that \eqref{eq:scal-weyl} agrees with \eqref{eq:grav-weyl} upon expanding the coefficients. 

% subsection grav-sc (end)

% section grav (end)

\section{Gravitational arguments for the Weyl law} % (fold)
\label{sec:grav-weyl}

\subsection{Average argument} % (fold)
\label{sub:ave}
One first strategy to use (\ref{eq:grav-weyl}) is to integrate both sides over all of $X$.
Physically, this corresponds to taking the four-dimensional particles to be mass distributions in the internal space instead of $D$-dimensional localized particles, a procedure sometimes called \emph{smearing}.
Equivalently for our purposes, a distribution in the internal space can also be thought of as a classical probability distribution assigned to a genuinely localized particle by a four-dimensional observer who is not able to probe the internal scales.\footnote{This point of view has been introduced in \cite{deluca-deponti-mondino-t-entropy}, which showed that, by associating an appropriate notion of entropy to such a probability distribution, a four-dimensional observer can reconstruct the internal geometry completely from thermodynamical quantities. Specifically, assigning a concavity property for this entropy is equivalent to assigning the complete set of internal Einstein equations.}
A four-dimensional observer maximally ignorant about the internal position of the particle would assign to it a uniform probability distribution in the internal space.
In both cases, this corresponds to integrating (\ref{eq:grav-weyl}) in the whole $X$ with a uniform probability distribution, which in curved space is a constant times $ \sqrt{g} $.

Taking $\int \dd^n y_0 \ee^{f(y_0)}  \sqrt{g_n}$ on both sides of (\ref{eq:scal-weyl}) and recalling the normalization (\ref{eq:psi-norm}), we obtain
\begin{equation}\label{eq:int-grav-weyl}
	\lim_{r\to 0} r^n\sum_{k=0}^\infty \ee^{-m_k r} \quad = \quad  \frac{n!\,\omega_n V(X)}{(2\pi)^n}\,,
\end{equation}
where we allowed the exchange of the limit and integral, a point to which we will come back at the end of this section.
Notice that thanks to the integration, we have gotten rid of the eigenfunctions and, simultaneously, all the dependence on the weight $f$ disappeared from \eqref{eq:int-grav-weyl}. Because of this, any result on the asymptotic behavior of the eigenvalues $m_k^2$ will be the same in the warped and unwarped case. 

We can now use a classic result on Laplace transforms due to Karamata and streamlined in \cite[XIII.5, Th.~2]{feller-vol2}. Applying it to the point measure on the eigenvalues and taking $m_k= \sqrt{\lambda_k}$ we obtain
\begin{equation}
	\lim_{\lambda \to \infty} \frac{N(\lambda)}{\lambda^{n/2}}= \frac{\omega_n V(X)}{(2\pi)^n}\,,
\end{equation}
which reproduces (\ref{eq:weyl-count}) with the coefficient $a$ given in (\ref{eq:weyl}).

For a rough idea of how this last step works, let us assume $m_k = \sqrt{\lambda_k} \sim \alpha k^{1/\nu}$ with two unknown constants $\alpha$ and $\nu$. When $r \ll 1$ the sum can be well approximated by an integral, using the general formula
\begin{equation}\label{eq:sum-int}
	\epsilon \sum_{k=0}^\infty f(k \epsilon) \underset{\epsilon\to 0}{\sim} \int_0^\infty \dd p f(p)\,.
\end{equation}
Taking $\epsilon=(\alpha r)^{\nu}$, we obtain
\begin{equation}\label{eq:weyl-grav}
	\sum_k \ee^{-m_k r} \underset{r\to 0}{\sim} \sum_k \ee^{-\alpha k^{1/\nu} r} \sim \frac1{(\alpha r)^{\nu}}\int_0^\infty \dd p \, \ee^{-p^{1/\nu}}=\frac{\Gamma(\nu+1)}{(\alpha r)^{\nu}} \,.
\end{equation}
Comparing with (\ref{eq:int-grav-weyl}), the power of $r$ determines $\nu=n$. The overall coefficient sets $\alpha$ to the $a$ in (\ref{eq:weyl}).

Summarizing, we have shown how to prove the Weyl law under the hypothesis that the exchange of the limit and integration that lead to \eqref{eq:int-grav-weyl} can be rigorously justified. In general, this requires to control the local geometry, in particular close to possible singularities. We will come back to this problem by providing an alternative proof of the Weyl law for D-brane type of singularities in Sec.~\ref{sec:weyl-rev}. In the next sections, we will discuss instead a local derivation of the Weyl law, which is thus not affected by these issues.
\subsection{Quantum ergodicity in the unwarped case} % (fold)
\label{sub:erg}

In our derivation in the previous subsection, integrating over the internal space was useful in getting rid of the eigenfunctions in (\ref{eq:grav-weyl}). For a large class of spaces $X$, an alternative strategy exists. 

A space $X$ is said to be (classically) \emph{ergodic} if, roughly said, for a generic choice of initial conditions, the geodesic motion of a particle covers uniformly all of phase space. Usually this is made precise as follows. Let $S^*X$ be the sphere bundle inside $T^*X$: it consists of all the choices of pairs $(x,\dot x)$, where $x\in X$ and $\dot x$ is the velocity vector (taken to have unit norm). It comes with a natural Liouville measure $\omega$.
A geodesic $\gamma$ parametrized by arc-length can be lifted canonically to the curve $(\gamma, \dot{\gamma})$ in $S^* X$. By definition, (classical) \emph{ergodicity} means
\begin{equation}\label{eq:erg}
	\lim_{T\to \infty}\frac1T \int_0^T f(\gamma(t), \dot\gamma(t)) \dd t= \int_{S^* X} f \dd \omega\,
\end{equation}
for almost all geodesics $\gamma$ and for any continuous $f$ on $S^* X$. In particular, one can take $f$ to be a function of the variable $x$ alone, i.e.~the pullback of a function on $X$. 

Not every $X$ has this property: when there are many isometries, geodesic motion might be integrable, and in that case (\ref{eq:erg}) will fail. As an extreme example, consider $\mathbb{S}^2$, with its usual coordinates $\theta$, $\phi$; there are three Killing vectors, and the associated conserved $L_i$ are the components of the angular momentum. A geodesic is a great circle: a particle will remain on it for ever, and not explore all of $\mathbb{S}^2$. In particular, if $L_z/L = \cos \theta_0$, the particle will always remain in a region around the equator:
\begin{equation}\label{eq:band-S2}
	\left[\frac\pi 2- \theta_0, \frac\pi 2+ \theta_0 \right]\,.
\end{equation}
So typical geodesics don't sample all of the $\mathbb{S}^2$, and (\ref{eq:erg}) fails for a non-constant function $f(\theta)$.

It is natural to wonder if an analogue of (\ref{eq:erg}) exists for a quantum particle on $X$. Indeed it turns out \cite{schnirelman,colindeverdiere,zelditch-surfaces} that if $X$ is ergodic, then it also has the \emph{quantum ergodicity} property, i.e.~the eigenfunctions $\psi_k$ of the Laplacian have the property that the probability to find the particle in a Borel subset $B\subset X$ is proportional to the measure of $B$:
\begin{equation}\label{eq:QE}
	\lim_{\substack{k\to \infty\\ k\notin e}}\frac{\int_B \sqrt{g}\psi_k^2}{ \int_{X} \sqrt{g}\psi_k^2} \ =\ \frac{V(B)}{V(X)}\,.
\end{equation}
The limit should be taken by possibly excluding a set $e\subset \mathbb{N}$ of $k$'s of measure zero: namely, $\lim_{k\to \infty}(\#(e\cap \{1,\ldots,k\}))/k =0 $. If $e=\emptyset$, then one says there is quantum \emph{unique} ergodicity. There has been extensive research over the years over this topic; the excluded set may not be empty, in which case the $\psi_{k\in e}$ are peaked along the non-generic, non-ergodic classical trajectories, called \emph{scars}. (See for example \cite{tao-scars} for an introduction and \cite[Sec.~7.5]{reichl} for some concrete examples with billiards.) 

Concretely, (\ref{eq:QE}) works because at large $k$ the $\psi_k$ oscillate wildly around a constant value; no matter how small $B$ is, at large enough $k$ taking the average of $\psi_k^2$ over it will suppress the oscillating part and give the constant. 

Just like its classical counterpart, QE fails in models with too many symmetries: taking again the round $\mathbb{S}^2$, the $\psi_k$ can be chosen to be the spherical harmonics $Y_{l,m}$ (the functions obtained by restricting polynomials on $\mathbb{R}^3$ to $\mathbb{S}^2$); recall that $\lambda=l(l+1)$, $|m|\leqslant l$. These will display the typical WKB behavior, with a wildly oscillating behavior in the classically allowed band (\ref{eq:band-S2}), and exponential decay outside it; again here $m/l= L_z/L= \cos \theta_0$. So even at large $\lambda$ the $\psi_k$ are not approximately constant. Of course if we sum over $m$ we do get a constant, $\sum_{m=-l}^l Y^2_{l,m}=(2l+1)/4\pi$. (The spectrum is invariant under the $L_i$, so the projector over a given level of $\Delta$ must be rotationally invariant.) Actually, even for this extremely symmetric case, if we select a random orthonormal basis of eigenfunctions rather than the $Y_{l,m}$, then quantum ergodicity does hold \cite{zelditsch-S2}. Essentially the random choice has the same effect as the sum over all $m$. More general results hold for random orthonormal bases (not necessarily of eigenfunctions) \cite{zelditch-QE-highdimension,maples}.

For the problem at hand, when ergodicity holds we can integrate (\ref{eq:grav-weyl}) over a small Borel subset $B$ rather than over all of $X$. The QE property was established for eigenfunctions of the ordinary unweighted Laplacian, so we still need to work with $f=0$. The $V(B)$ factors cancel out, and we reduce to (\ref{eq:int-grav-weyl}) again, from which point our previous computation applies. 

This argument for Weyl's law has the advantage that it doesn't require integrating over the whole $X$. This is physically more meaningful, because it corresponds to particles that are localized in the internal space; and it avoids possible issues when $X$ has singularities.\footnote{In string theory, solutions with singularities have $A\neq 0$, for which we actually need to use the argument in the next subsection.}
The disadvantage is that it only applies to ergodic spaces. While many of the best known solutions have lots of symmetries, one expects the ergodic case to be the generic one.

% subsection erg (end)

\subsection{Weighted quantum ergodicity} % (fold)
\label{sub:wqe}

We will now consider the case of warped compactifications, i.e.~$f\neq 0$. As we have seen in the averaged derivation in Sec.~\ref{sub:ave}, Weyl's law contains the ordinary volume, while (\ref{eq:grav-weyl}) contains the weighted volume (\ref{eq:wV}), ultimately because of the relation (\ref{eq:GG}) among the Planck masses. When integrating (\ref{eq:grav-weyl}) over the internal space, this cancels, by virtue of the left hand side depending on $\ee^{f(y_0)}$, leaving just an ordinary volume. 
Trying to avoid the internal integral, however, a puzzle arises here. Quantum ergodicity would suggest that $\psi_k(y_0)^2\to 1$, but this would not cancel the weighted volume factor on the right hand side of (\ref{eq:grav-weyl}) nor transform it into an ordinary volume. What reconciles the Weyl law with the local behavior of (\ref{eq:grav-weyl})?

One clue to the solution is that the left hand side of (\ref{eq:grav-weyl}) depends on $y_0$, while the right hand side does not.  This implies that the wavefunctions cannot possibly oscillate around a constant at large $k$; rather, 
\begin{equation}\label{eq:wqe}
	\psi_k^2(y_0) \ \text{oscillate around} \ \propto\ee^{-f(y_0)}\,.
\end{equation}

Notice that this is the \emph{inverse} of the function appearing in the definition of the weighted volume (\ref{eq:wV}). When we integrate \eqref{eq:wqe} over a Borel subset $B\subset X$ with that measure, we thus end up with an ordinary $V(B)$ on the right-hand side:
\begin{equation}\label{eq:WQE-norm}
	\lim_{\substack{k\to \infty\\ k\notin e}}\int_B \sqrt{g}\ee^f\psi_k^2\ = \  V_f(X)\frac{V(B)}{V(X)} \,,
\end{equation}
where we fixed the normalization constant by consistency with the case $B =X$ in our normalization (\ref{eq:psi-norm}).
For a more general normalization, we have
\begin{equation}\label{eq:WQE}
	\lim_{\substack{k\to \infty\\ k\notin e}}\frac{\int_B \sqrt{g}\ee^f\psi_k^2}{\int_{X} \sqrt{g}\ee^f\psi_k^2}\ = \  \frac{V(B)}{V(X)}\quad \text{for every Borel subset $B\subset X$} \,.
\end{equation}
It is natural to call (\ref{eq:WQE}) the \emph{weighted quantum ergodicity} (WQE) property.  It might have applications independent of the Weyl law. While the QE in the previous subsection is proven in the classically ergodic case, the corresponding statement for (\ref{eq:WQE}) is at this point a conjecture. We give in App.~\ref{sub:ex-erg} a couple of simple one-dimensional examples where it can be checked rather explicitly. 

Integrating  (\ref{eq:grav-weyl}) over a small ball where  the particles are present, and using  (\ref{eq:WQE-norm}) gives the Weyl law  (\ref{eq:weyl}). Let us stress that  the WQE  is key in order to cure the naive discrepancy between the gravitational expectation \eqref{eq:grav-weyl}  and the Weyl law (\ref{eq:weyl}), as in the former the weighted volume $V_f(X)$ appears  while in the latter only the classical volume $V(X)$ shows up.  We give below the precise statement, followed by  a proof.

\begin{theorem}\label{Th:WQEToWeyl}
Let $(X,g, e^f)$ be a weighted Riemannian $n$-dimensional manifold, possibly with isolated singularities. Assume that 
\begin{enumerate}
\item[(i)] The Riemannian volume $V(X)$ and the weighted volume $V_f(X)$ are finite;
\item[(ii)]  The spectrum of the weighted Laplacian $\Delta_f$ is discrete and satisfies 
\begin{equation}\label{eq:limsuprto0}
\limsup_{r\to0} r^n\sum_{k=0}^\infty \ee^{-\sqrt{\lambda_k} r}<\infty\,;
\end{equation}
\item[(iii)]   The weighted quantum ergodicity property \eqref{eq:WQE} holds, with the excluded set $e$ of finite cardinality.
\end{enumerate}
 Then the eigenvalues of the weighted Laplacian $\Delta_f$ satisfy the (regular) Weyl law  \eqref{eq:weyl}. 
\end{theorem}

\begin{proof}
 Fix  $y_0\in X$,   a regular point for both the Riemannian $n$-dimensional manifold $(X,g)$ and the weight $f$. Fix $\varepsilon>0$ and let  $B$ be a small metric ball of radius $r$ centred at $y_0$. To keep notation short, we denote $m_k:=\sqrt{\lambda_k}$.
Integrating \eqref{eq:scal-weyl} (or, equivalently, \eqref{eq:grav-weyl}) over $B$ with respect to the Riemannian volume  $\sqrt{g}$, and using that the series $ \sum_{k=0}^\infty \ee^{-m_k r} \psi_k^2$ converges in $L^1(M, e^f {\rm vol}_g)$, we get that
\begin{equation}\label{eq:grav-weyl-Pf}
    \left| \frac{n!\omega_n V(X) }{(2\pi)^n} - \frac{V(X)}{ V_f(X) V(B)} r^n \sum_{k=0}^\infty  \ee^{-m_k r} \int_B \sqrt{g} e^{f} \psi_k^2   \right| \leqslant \frac{\varepsilon}{4},
\end{equation}
for all $r\in (0,  r_0)$, where $r_0= r_0(\varepsilon)>0$ is small enough.
Notice that the weighted quantum ergodicity property \eqref{eq:WQE} together with the  normalization (\ref{eq:psi-norm}) gives (\ref{eq:WQE-norm}). Now, using (\ref{eq:WQE-norm})  under the additional assumption that the excluded set $e$ has finite cardinality, we infer that there exists $N=N(\varepsilon')>0$ such that
\begin{equation}\label{eq:WQE-norm-Pf}
\left|  \  V_f(X)\frac{V(B)}{V(X)} -  \int_B \sqrt{g}\ee^f\psi_k^2\right| \leqslant \varepsilon', \quad \text{for all } k\geqslant N.
\end{equation}
Notice also that
\begin{equation}\label{eq:finiteSum1}
 \frac{V(X)}{ V_f(X) V(B)}  r^n  \sum_{k=0}^N  \ee^{-m_k r} \int_B \sqrt{g} e^{f} \psi_k^2   \leqslant \frac{\varepsilon}{4},  \quad  \text{ for all }r\in (0,  r_0(\varepsilon)).
\end{equation} 
Plugging \eqref{eq:WQE-norm-Pf}--\eqref{eq:finiteSum1}  into  \eqref{eq:grav-weyl-Pf} and recalling \eqref{eq:limsuprto0}, we obtain
\begin{equation}\label{eq:grav-weyl-Pf2}
    \left| \frac{n!\omega_n V(X) }{(2\pi)^n} -   r^n \sum_{k=N}^\infty  \ee^{-m_k r}    \right| \leqslant \frac{\varepsilon}{2}, \quad \text{for all }r\in (0,  r_0(\varepsilon)).
\end{equation}
Notice also that
\begin{equation}\label{eq:finiteSum2}
r^n \sum_{k=0}^N  \ee^{-m_k r}   \leqslant \frac{\varepsilon}{2},   \quad \text{ for all }r\in (0,  r_0(\varepsilon)).
\end{equation} 
Splitting  the series in \eqref{eq:grav-weyl-Pf} as the finite sum up to $N$ and the infinite series for $k\geqslant N$, and using \eqref{eq:grav-weyl-Pf2}--\eqref{eq:finiteSum2} yields
\begin{equation}\label{eq:grav-weyl-Pf3}
    \left| \frac{n!\omega_n V(X)}{(2\pi)^n} -  r^n \sum_{k=0}^\infty  \ee^{-m_k r}    \right| \leqslant \varepsilon,  \quad  \text{ for all }r\in (0,  r_0(\varepsilon)).
\end{equation}
Since $ \varepsilon>0$ is arbitrarily small,  \eqref{eq:grav-weyl-Pf3} leads us back to (\ref{eq:int-grav-weyl}), which we have already shown to imply the Weyl law \eqref{eq:weyl}. 
\end{proof}

\begin{remark}
In Th.\;\ref{Th:WQEToWeyl}, in place of (iii) it is sufficient to ask the following (a priori weaker) variant of the  WQE: there exists a regular point $y_0\in X$ such that \eqref{eq:WQE} holds for a sequence of metric balls $B=B_{r_j}(y_0)$ around $y_0$ whose radii $r_j$ converge to $0$. 
\end{remark}

\begin{remark}
It is natural to expect that the finite cardinality condition in Th.\;\ref{Th:WQEToWeyl} (iii) can be dropped. The contribution of some of the $e$ can be large if $B$ intersects one of the scars (i.e.~the classical closed trajectories; recall the discussion below (\ref{eq:QE})). But the number of scars intersecting a single $B$ is unlikely to be large.
\end{remark}

For another heuristic argument leading to the Weyl law, at least in the smooth case, we can use a standard trick and map the original eigenvalue problem $\Delta_f \psi_k = \lambda_k \psi_k$ for (\ref{eq:w-Lap}) to a Schr\"odinger equation:
\begin{equation}\label{eq:schroedinger}
	 \left(\Delta_0 + U\right)\tilde\psi_k = \lambda_k \tilde\psi_k \, ,\qquad U= -\ee^{-f/2}\Delta_0 \,\ee^{f/2} \, ,\qquad \psi_k=\ee^{-f/2}\tilde\psi_k\,,
\end{equation}
where $\Delta_0$ is now the ordinary unweighted Laplace--Beltrami operator on $X$. If $|U|$ is bounded, the $\lambda_k$ at large $k$ should be much larger than $\mathrm{sup}_X U$. In other words, at large $k$ the $\tilde\psi_k$ should oscillate fast, so the kinetic term in (\ref{eq:schroedinger}) should dominate over $U$.\footnote{
Note that $U$ is in fact not always bounded for the singularities that appear in string theory. In the usual coordinates (see e.g.~\cite[Sec.~3]{deluca-deponti-mondino-t}), for D$p$-branes $U\sim-1/16(p-3)(p-7)^3r^{5-p}$, which as $r\to 0$ diverges for $p=6$. For O$p$-planes, $U\sim - 1/16 (p-3)(p-7) (r-r_0)^{-3}$, whose $r\to r_0$ limit is $+\infty$ for $p=4,5,6$ and $-\infty$ for $p=8$. In addition, even in the cases where $U$ is bounded, the geometry itself is singular for various sources, and this heuristic argument would not apply directly since quantum ergodicity for $\Delta_0$ might not hold. We thank Zhenbin Yang for discussions on this point.} 
So the $\tilde\psi_k$ should be asymptotic at large $k$ to the eigenfunctions of $\Delta_0$. If ergodicity holds, we know by the previous subsection that the $\tilde \psi_k$ will oscillate around a constant. By the rescaling in (\ref{eq:schroedinger}), we can now conclude (\ref{eq:wqe}).

It should be noted that the WQE is not expected to hold for a space that has symmetries, just like its usual unweighted counterpart. These are typically the models where the spectrum can be solved explicitly. In App.~\ref{sub:ex-nerg} we will see an example where the WQE does not hold, but the spectrum is so explicit that the gravitational expectation (\ref{eq:grav-weyl}) can be checked directly.

% subsection wqe (end)

\subsection{Possible application to gravity localization} % (fold)
\label{sub:loc}

The WQE property  (\ref{eq:WQE}) might also have an application to gravity localization. This is the idea that some warped product spacetimes might give rise to a gravitational potential that behaves   in the four-dimensional way $\propto 1/r$ over some range of distances, even if the ``internal'' space $X$ is non-compact. It was famously realized in the Randall--Sundrum II \cite{randall-sundrum2} and Karch--Randall models \cite{karch-randall}, for Minkowski and AdS spacetimes respectively, with string theory realizations explored for example in \cite{verlinde-RS2,chan-paul-verlinde,bachas-lavdas2}. 

The idea involves often a mix of two phenomena: i) a separation between the first spin-two KK mass and the rest of the tower, and ii) a suppression of the wavefunctions for the latter. The first can be analyzed mathematically using our work on estimates on the KK tower \cite{deluca-deponti-mondino-t,deluca-deponti-mondino-t-entropyproof,deluca-deponti-mondino-t-L2}. The second is harder, as there appears to be very little mathematical literature on eigenfunctions of the weighted Laplace operator, especially in the presence of singularities. 

Our result (\ref{eq:wqe}) about the eigenfunctions, albeit not fully rigorous, provides a partial remedy to this. If $\ee^f= \ee^{(D-2)A}$ is peaked in a region, the $\psi_k$ at large $k$ will be suppressed there. Ideally one would of course also want results about small $k$.

For example, in earlier work \cite[Sec.~4]{deluca-deponti-mondino-t-L2} we found examples of AdS warped products where eigenvalue separation is realized, estimating the $m_k$ using Cheeger constants. However, this result alone guaranteed a four-dimensional gravitational behavior only for distances larger than the cosmological scales; we found that lowering this scale would be possible only by using the aforementioned wavefunction suppression. 

Let us further restrict to the models of \cite[Sec.~4.3]{deluca-deponti-mondino-t-L2} with ${\mathcal N}=4$ supersymmetry, previously studied in this context by Bachas, Estes and Lavdas \cite{bachas-estes,bachas-lavdas2}. Here the non-compact internal space has a central ``bulb'' connected via two thin ``tubes'' to two non-compact ends. The function $\ee^f$ has a peak in the bulb, is small in the tubes, and grows again exponentially in the non-compact ends. So (\ref{eq:wqe}) leads to expect that the $\psi_k$ are localized on the tubes for large $k$. A numerical study suggests that this indeed happens. Further more general results on lower $k$ would be needed to really conclude that the scale of localization can be made much lower than the cosmological scale. But we find the results obtained here from WQE to be an encouraging step in that direction.

% subsection loc (end)

% section grav-weyl (end)

\section{Weyl law and RCD spaces}\label{sec:weyl-rev}
We discuss here the mathematical literature about the Weyl law for non-smooth spaces, with a particular focus on the $\mathsf{RCD}$ setting. We then provide a rigorous proof of the validity of the Weyl law for compact spaces with $Dp$-brane singularities, for $p=6,7,8$.

\subsection{Mathematical results in the non-smooth setting}
\label{sub:mathrev}
In the last decades, there has been a tremendous interest in the theory of curvature-dimension bounds on non-smooth spaces (see \cite{Amb} for a survey). Among the various conditions introduced and studied, a prominent role is played by the class of $\mathsf{RCD}(K,N)$ spaces, where the Weyl law was recently investigated. This class consists of metric measure spaces $(X,\di,\mm)$ with \emph{synthetic} Ricci curvature bounded below by $K\in \mathbb{R}$, and dimension bounded above by $N\in [1,\infty).$\footnote{The $\mathsf{RCD}(K,N)$ class can be also defined for $N=\infty$ and $N<0$. Weyl law in these more general settings has not been considered yet in the literature.} These spaces may possess singularities, but it was proved in \cite{mondino-naber} that the regular part covers the whole space, possibly up to a singular set of $\mm$-measure zero. Roughly, one defines the regular part as the set of points such that the metric measure space looks like the Euclidean space at an infinitesimal scale. More precisely, for every $n\in \mathbb{N}$ let us introduce the set $\mathcal{R}_n$ of points $x\in \mathrm{supp}(\mm)$ such that
\begin{equation*}
(X,r^{-1}\di, \mm^x_r,x)\xrightarrow{\text{pmGH}} (\mathbb{R}^n,|\cdot|,c_n\mathcal{H}^n,\mathrm{0}_n) \qquad \textrm{as} \ r\to 0^+\,.
\end{equation*}
Here $|\cdot|$ is the Euclidean distance, $\mathcal{H}^n$ is the Hausdorff measure (or in other words the standard Lebesgue volume measure), and $\mm^x_r$ is the rescaled measure defined as
$$\mm^x_r:=\left(\int_{B_r(x)}(1-r^{-1}\di(x,y))\dd \mm(y)\right)^{-1}\mm\,.$$
$c_n$ is the natural normalization constant and by $\textrm{pmGH}$ we are denoting the pointed measured Gromov--Hausdorff convergence (see \cite{mondino-naber} for all the relevant definitions). 
It is known thanks to \cite{brue-semola} that there exists an $n\in [1,N]$, called \emph{essential dimension} of the space $(X,\di,\mm)$, such that $\mm(X\setminus \mathcal{R}_n)=0$. Introducing the reduced regular set $\mathcal{R}^*_n$ defined as
$$\mathcal{R}^*_n:=\left\{x\in \mathcal{R}_n \ : \ \exists \lim_{r\to 0^+} \frac{\mm(B_r(x))}{w_nr^n}\in (0,\infty)\right\},$$
it is proved in  \cite{ambrosio-honda-tewodrose} that on any compact $\mathsf{RCD}(K,N)$ space we have $\mm(\mathcal{R}_n\setminus \mathcal{R}^*_n)=0$ and the \emph{regular} Weyl law  
$$\lim_{\lambda \to \infty}\frac{N(\lambda)}{\lambda^{n/2}}=\frac{w_n}{(2\pi)^n}\mathcal{H}^n(\mathcal{R}^*_n)$$
holds if and only if
\begin{equation}\label{eq: AHT condition}
 \lim_{r\to 0^+} \int_X \frac{r^n}{\mm(B_r(x))}\dd \mm=\int_X \lim_{r\to 0^+} \frac{r^n}{\mm(B_r(x))}\dd \mm<\infty\,.
\end{equation}
We remark that in \eqref{eq: AHT condition} it is required both the equality between the two integral expressions, and the finiteness of them.
An important class where these conditions are always satisfied (with $n=N$) is the one of compact \emph{non-collapsed} $\mathsf{RCD}(K,N)$ spaces, namely those $\mathsf{RCD}(K,N)$ spaces where $\mm=\mathcal{H}^N$. The validity of the regular Weyl law was also established for the Dirichlet problem in a bounded domain inside the ambient $\mathsf{RCD}(K,N)$ space $(X,\di,\mathcal{H}^N)$ \cite{zhang-zhu}, and for general compact $\mathsf{RCD}(K,N)$ spaces with essential dimension equal to $1$ \cite{IwKiYo23}.

The situation may drastically change if one works in higher dimensions and removes the non-collapsed assumption: it was shown in \cite{DHPW23} that for any $\beta\in (2,\infty)$ there exist compact $\mathsf{RCD}(-1,N)$ spaces of essential dimension $2$ such that 
$$0<\lim_{\lambda \to \infty}\frac{N(\lambda)}{\lambda^{\beta/2}}<\infty\,.$$  
Moreover, in the same class of spaces the asymptotic regime of the eigenvalues can also exhibit a logarithmic correction of the form
$$\lim_{\lambda \to \infty}\frac{N(\lambda)}{\lambda\log{\lambda}}=\frac{1}{4\pi}\,.$$  
We refer to \cite{ChPrRi24} for some other interesting results concerning \emph{non-regular} Weyl laws in the presence of singularities, where the spaces are non-complete Riemannian manifolds $(\mathbb{M},g)$ with the curvature blowing up when approaching the metric boundary, and to \cite{pesenson20} where, for general metric measure spaces satisfying some weak regularity conditions, it is shown that the number $N(\lambda)$ is related to the cardinality of appropriate covers of the space with balls of radius $\lambda^{-1/2}$.

We proved in \cite{deluca-deponti-mondino-t} that spaces whose only singularities are of D-brane type are $\RCD$. They are not non-collapsed, nonetheless we will check in Sec.\;\ref{sub:mathproofDbrane} that (\ref{eq: AHT condition}) is satisfied for D$p$-branes for $p\geqslant 6$; thus the Weyl law is valid. Notice that D$p$-brane singularities with $p\geqslant 6$ are exactly those for which the spectrum is discrete, as we will clarify in Prop.\;\ref{prop:Dbranediscrete}. On the other hand, we showed in \cite[Sec.~6.4.1]{deluca-deponti-mondino-t-entropy} that spaces with O-plane singularities are not in the $\RCD$ class so we cannot conclude from the $\mathsf{RCD}$ theory that the Weyl law holds for compactifications with O-planes. Our physics argument above, on the other hand, although do not constitute a complete mathematical proof, covers such spaces as well.

% subsection mathrev (end)

\subsection{Validity of Weyl law for singularities of D-brane type via RCD theory}
\label{sub:mathproofDbrane}
Here we prove that exact D$p$ branes satisfy \eqref{eq: AHT condition} and thus the regular Weyl law for $p=6,7,8$,  thanks to  \cite{ambrosio-honda-tewodrose}. Since it will be relevant for the following computations, we start by recalling the precise definition of the space we will treat, focusing for simplicity on the case with $d=4$ uncompactified dimensions. We refer the interested reader to \cite{deluca-deponti-mondino-t} for the definition of spaces with D$p$ branes singularities, $p\leqslant 5$.

\begin{definition}[Exact D-brane metric measure spaces]\label{def: Dp-brane mms}
We define an \emph{exact Dp-brane metric measure space} a smooth and compact Riemannian manifold $(X,g)$ that is glued (in a smooth way) to a finite number of ends where the metric has a precise form that we specify below. 
The distance $\di$ between two points $p,q \in X$ is given by 
$$\di(p,q):=\inf_{\gamma\in \Gamma(p,q)} \int g\left(\gamma'(t),\gamma'(t)\right) \dd t\,,$$
where $\Gamma(p,q)$ denotes the set of absolutely continuous curves joining $p$ to $q$.

The measure $\mm$ is a weighted volume measure $\mm:=\ee^f\mathsf{\dd vol}_g$, where $ \mathsf{\dd vol}_g$ is the Riemannian volume measure associated to $g$ and the function $\ee^f$ is smooth outside the singular set and gives zero mass to it.

Depending on the value of $p$, near the singular set the metric $g$ and the measure satisfy:
\begin{itemize}
\item Case $p=6$. In a neighborhood $\{r<\epsilon\}$ of the singular set $\{r=0\}$,  the metric is of the form  \begin{equation}\label{eq:metricEndAs1}
	g=\dd y^2_{3}+\left( \frac{r_0}{r} \right) \left( \dd r^2+ r^2 \dd s^2_{\mathbb{S}^{2}} \right)\,,
	\end{equation}
with $r_0= g_s (2\pi l_s)/\mathrm{Vol}(\mathbb{S}^{2})$ and $\dd y^2_{3}$ being the flat metric of a $3$-dimensional torus. The measure is given by 
	$$
	\mm\llcorner_{\{r<\epsilon \}}= \sqrt{\frac{r}{r_0}} \, \mathsf{\dd vol}_g \llcorner_{\{r<\epsilon \}}\,.
	$$
        
        \item Case $p=7$. In a neighborhood $\{r<\epsilon\}$ of the singular set $\{r=0\}$,  the metric is of the form
        \begin{equation}\label{eq:metricEndAs7}
	g= \dd y^2_{4}- \frac{2\pi}{g_s}\log (r/r_0) \left( \dd r^2+ r^2 \dd s^2_{\mathbb{S}^{1}} \right)\,,
	\end{equation}
	with $r_0>0$ a constant and $\dd y^2_{4}$ being the flat metric of a $4$-dimensional torus. The measure is given by 
	$$
	\mm\llcorner_{\{r<\epsilon \}}=  \mathsf{vol}_g \llcorner_{\{r<\epsilon \}}\,.
	$$
	\item      Case $p=8$. In a neighborhood $\{|r|<\epsilon\}$ of the singular set $\{r=0\}$,  the metric is of the form
        \begin{equation}\label{eq:metricEndAs8}
	g=\dd y^2_{5}+ (1- h_8 |r|) \dd r^2\,, 
	\end{equation}
	with $h_8>0$ a constant and $\dd y^2_{5}$ being the flat metric of a $5$-dimensional torus. The measure  is given by 
	$$
	\mm\llcorner_{\{|r|<\epsilon \}}= \sqrt{1- h_8 |r|} \, \mathsf{\dd vol}_g \llcorner_{\{|r|<\epsilon \}}\,.
	$$
\end{itemize}
\end{definition}

In the next proposition, we collect some results obtained in our previous works.

\begin{proposition}\label{prop:Dbranediscrete}
An exact Dp-brane metric measure space is a $\mathsf{RCD}(K,N)$ space for some $K\in \mathbb{R}$ and $N\in [1,\infty)$ with essential dimension $n=6$. Moreover:
\begin{itemize}
\item It is a compact metric space and the spectrum is discrete if $p=6,7,8$.
\item The spectrum is not discrete if $p<5$. More precisely, there are no positive eigenvalues below the infimum of the essential spectrum of the Laplacian. 
\end{itemize} 
\end{proposition}
\begin{proof}
The validity of the $\mathsf{RCD}(K,N)$ condition was proved in \cite[Th.~3.2]{deluca-deponti-mondino-t}. Since the space is smooth outside the singular set, it is immediate that the essential dimension of the space coincides with the dimension of the underlying manifold and thus $n=6$ (recall here that we are assuming $d=4$ uncompactified dimensions). It was also checked in \cite{deluca-deponti-mondino-t} that a D$p$-brane metric measure space is compact if and only if $p=6,7,8$, and thus for these values of $p$ the spectrum is discrete as a consequence of \cite[Th.~6.3]{GMS}. For $p<5$ it was noticed in \cite[Sec.~4.2.2]{deluca-deponti-mondino-t} that $h_1=0$ since for tubular neighborhoods $B_R$ of the singular set we have $\mathsf{Per}(B_R)/{\mm(B_R)}\to 0$ as $R\to 0$; here, $h_1$ is the Cheeger constant of the space. The last conclusion of the proposition is thus a consequence of \cite[Th.~4.2]{deluca-deponti-mondino-t} (based on the Buser inequality proved in \cite{deponti-mondino}). 
\end{proof}
Although we have not detailed a completely rigorous mathematical argument, we remark that the spectrum is expected to have a continuous part also for D$5$ brane singularities: this is suggested by i) a local study of the eigenvalue equation, ii) a study of the Cheeger constants \cite[Sec.~4.2.2]{deluca-deponti-mondino-t}, an estimate \cite[Sec.~4.1]{deluca-deponti-mondino-t-L2} of capacities based on \cite{cianchi-mazya}. A difference with the cases $p<5$ is that for D$5$-brane metric measure spaces the Cheeger constant is expected to be positive. By similar considerations based on an analysis of tubular neighborhoods of the singularities, we also find that the spectrum of O$p$-plane singularities appears to be discrete.\footnote{This is indeed the case for example for the explicit computation in \cite{passias-richmond}, where an O8 is present.}

We are now ready to state and prove the main result of the section.
\begin{theorem}\label{th: checkDpweyl}
An exact Dp-brane metric measure space satisfies \eqref{eq: AHT condition} (and thus the regular Weyl law) for $p=6,7,8$.
\end{theorem}
\begin{proof}
We start by proving the right hand side inequality in \eqref{eq: AHT condition}. For this purpose, we notice that the set of singularities is $\mm$-negligible and for any other point $x\in X\setminus\{r=0\}$ we have 
\[
\lim_{R\to 0^+} \frac{R^{6}}{\mm(B_R(x))}=\frac{C}{e^{f(x)}}
\]
for a constant $C$ depending only on $r_0$ or $h_8$. This yields
\[
\int_X \lim_{R\to 0^+} \frac{R^{6}}{\mm(B_R(x))}\dd \mm=C\; \mathsf{vol}_g(X\setminus\{r=0\})<\infty\,.
\]

We are left to show the equality
\begin{equation}\label{eq:claimR6mmBR}
 \lim_{R\to 0^+} \int_X \frac{R^{6}}{\mm(B_R(x))}\dd \mm=\int_X \lim_{R\to 0^+} \frac{R^{6}}{\mm(B_R(x))}\dd \mm
\end{equation}
which follows by applying the dominated convergence theorem as we specify below by handling the three different cases.

\textbf{Case} $p=6$. In the following we always suppose to work in the set $\{0<r<\epsilon\}$, which is the only which needs a discussion (the complement is smooth and compact). With the change of variable $\rho=2\sqrt{r_0}\sqrt{r}$, the metric takes the form
$$g=\dd y^2_{3}+\left( \dd \rho^2+ \frac{\rho^2}{4} \dd s^2_{\mathbb{S}^{2}} \right),$$
and the measure is given by 
	$$
	\mm= \frac{\rho}{2r_0} \, \mathsf{\dd vol}_g =\frac{\rho^3}{8r_0}\sin{\theta}\,\dd\rho\,\dd\theta\,\dd\phi\,\dd^{3}y,
	$$
 where $\theta\in [0,\pi]$ and $\phi\in [0,2\pi)$ are the spherical variables in ${\mathbb{S}^{2}}$.
We fix $\overline{R}>0$ and our aim is to find $w\in L^1(\mm)$ such that 
$$\sup_{0<R<\overline{R}} \frac{R^{6}}{\mm(B_R(x))}\leqslant w(x) \qquad \textrm{for} \ \mm-a.e. \ x.$$ 
To compute $\mm(B_R(x))$, we start by noticing that this quantity only depends on $R$ and $\rho_0:=\di(x,\{r=0\})$. Without loss of generality, we can thus suppose \[
x=(y=\underline{0},\rho=\rho_0,\theta=0,\phi=\phi_0)\in X\setminus\{r=0\}\,.
\]
Expressing the flat space in polar coordinates with radial coordinate $\varrho$, we have
\begin{equation}\label{eq:exprmeasballproof}
\mm(B_R(x))=4\pi\int_0^R\varrho^{2}I\left(\rho_0,\sqrt{R^2-\varrho^2}\right)\,\dd\varrho\,,
\end{equation}
with
\begin{align*}
I(\rho_0,\rho_1):&=\int_{\tilde{B}_{\rho_1}(x)}\frac{\rho^3}{8r_0}\sin{\theta}\,\dd\rho\,\dd\theta\,\dd\phi=\frac{\pi}{2r_0}\int_{\max\{\rho_0-\rho_1;0\}}^{\rho_0+\rho_1}z'\left(\rho_1^2-(z'-\rho_0)^2\right)\dd z'\\
&=\begin{cases}
    \frac{2\pi}{3r_0}\rho_0\rho_1^3 \qquad &\textrm{if } \rho_0\geqslant \rho_1\\
    \frac{\pi}{24r_0}(3\rho_1-\rho_0)(\rho_0+\rho_1)^3 \qquad &\textrm{if } \rho_0< \rho_1\,,
\end{cases}
\end{align*}
where $\tilde{B}_{\rho_1}(x)$ denote the geodesic ball in the $3$-dimensional factor with metric $\tilde{g}=\dd \rho^2+ \frac{\rho^2}{4} \dd s^2_{\mathbb{S}^{2}}$, and the expression of $I$ is motivated by the fact that, up to the change of variables $\theta'=\theta/2\in [0,\pi/2]$, this geodesic ball corresponds to the $3$-dimensional Euclidean ball.
The integral in \eqref{eq:exprmeasballproof} can be computed explicitly: for $R\leqslant  \rho_0$ the expression is given by
\[
\mm(B_R(x))=\frac{8\pi^2\rho_0}{3r_0}\int_0^{R}\varrho^2(R^2-\varrho^2)^{\frac{3}{2}}\dd\varrho=\frac{\pi^3}{12r_0}\rho_0R^6\,;
\]
when $R> \rho_0$ we get instead the following expression
\begin{align*}
\mm(B_R(x))=\frac{8\pi^2\rho_0}{3r_0}&\int_0^{(R^2-\rho_0^2)^{\frac{1}{2}}}\varrho^2(R^2-\varrho^2)^{\frac{3}{2}}\dd\varrho\\
&+\frac{\pi^2}{6r_0}\int_{(R^2-\rho_0^2)^{\frac{1}{2}}}^R\varrho^2\left(3(R^2-\varrho^2)^{\frac{1}{2}}-\rho_0\right)\left(\rho_0+(R^2-\varrho^2)^{\frac{1}{2}}\right)^3\dd\varrho
\end{align*}
\begin{align*}
=\frac{\pi^2}{1260r_0}\bigg(105\rho_0R^6\left(\pi-\arccos\left(\frac{\rho_0}{R}\right)\right)+(R^2-\rho_0^2)^{\frac{1}{2}}(48R^6+87\rho_0^2R^4-38\rho_0^4R^2+8\rho_0^6)\bigg).
\end{align*}

Since $48R^6+87\rho_0^2R^4-38\rho_0^4R^2+8\rho_0^6\geqslant 0$ and $\pi-\arccos(\rho_0/R)\geqslant \pi/2$, it follows 
\[
\sup_{0<R<\overline{R}} \frac{R^{6}}{\mm(B_R(x))}\leqslant \frac{\tilde{C}}{\rho_0}
\]
for a suitable constant $\tilde{C}=\tilde{C}(r_0, \overline{R})$ independent on $\rho_0$. The proof is concluded by taking 
\[
w(x):=\frac{\tilde{C}}{\di(x,\{r=0\})}\in L^1(\mm).
\]

\vspace{3mm}
 For the cases $p=7,8$, the direct computations are more heavy. It is therefore useful to perform some simplifications.  Note we can neglect the flat part given by the Euclidean metric $\dd x^2_{p+1-d}$ since it is smooth and with finite diameter (and thus with finite volume). 
We denote by $\mathsf{O}:=\{r=0\}$ the singular point. Near $\mathsf{O}$, the Riemannian metric and the measure are of the form  
\begin{equation}\label{eq:gmH}
\bar{g}=H(r)(\dd r^2+r^2 \dd s^2_{\mathbb{S}^{8-p}}), \quad \bar{\mm}=H(r)^{\frac{p-7}{2}}\dd\mathsf{vol}_{\bar{g}},
\end{equation}
with $H(r)=-\frac{2\pi}{g_s}\log(r/r_0)$ for $p=7$, and $H(r)=1-h_8|r|$ for $p=8.$

\textbf{Case} $p=7$.
Thanks to the dimension reduction discussed above,  \eqref{eq:claimR6mmBR} reduces to show that, for some $\varepsilon_0>0$,
\begin{equation}\label{eq:claimR2mmBR}
 \lim_{R\to 0^+} \int_{B^{\bar{g}}_{\varepsilon_0}(\mathsf{O})} \frac{R^{2}}{\bar \mm(B_R(x))}\dd \bar{\mm}=\int_{B^{\bar{g}}_{\varepsilon_0}(\mathsf{O})} \lim_{R\to 0^+} \frac{R^{2}}{\bar \mm(B_R(x))}\dd \bar\mm.
\end{equation}
Note that \eqref{eq:claimR2mmBR} is a consequence of the following two claims: 
\begin{align}
&\text{Claim 1:} \quad   \lim_{R\to 0^+} \int_{B^{\bar{g}}_{\varepsilon_0}(\mathsf{O})\setminus B^{\bar{g}}_{2R}(\mathsf{O}) } \frac{R^{2}}{\bar \mm(B_R(x))}\dd \bar{\mm}=\int_{B^{\bar{g}}_{\varepsilon_0}(\mathsf{O})} \lim_{R\to 0^+} \frac{R^{2}}{\bar \mm(B_R(x))}\dd \bar\mm. \label{Claim1-p7}\\
&\text{Claim 2:} \quad   \lim_{R\to 0^+} \int_{B^{\bar{g}}_{2R}(\mathsf{O}) } \frac{R^{2}}{\bar \mm(B_R(x))}\dd \bar{\mm}= 0. \label{Claim2-p7}
\end{align} 

We start by showing the first claim \eqref{Claim1-p7}. %Up to constants (which are not relevant for the arguments), $H(r)=-\log(r)$. Thus
We denote by $C(\mathbb{S}^1)$ the metric cone over $\mathbb{S}^1$ with metric $\dd r^2+r^2 \dd s^2_{\mathbb{S}^{1}}$.
From \eqref{eq:gmH}, we have
\begin{equation}\label{eq:mmLB}
 \bar{\mm}\llcorner B^{\bar{g}}_R(x)\geqslant \inf_{B^{\bar{g}}_R(x)} H\;  \dd\mathsf{vol}_{r}\,\dd\mathsf{vol}_{\mathbb{S}^1}
\end{equation}
and
\begin{equation}\label{eq:ballInclusions}
B^{\bar{g}}_R(x)\supset B^{C(\mathbb{S}^1)}_{\frac{R} {\sup_{B^{\bar{g}}_R(x)} \sqrt{H}}}(x).
\end{equation}
The $\bar{g}$-distance $\rho$ between $\mathsf{O}$ and a point of coordinate $r$ is given by (up to a multiplicative constant that is not relevant for this argument)
\begin{equation}\label{eq:rtorho}
\rho=\int_0^r \sqrt{-\log(s)}\,\dd s\,.
\end{equation}
It is easy to check that
\begin{equation}\label{eq:rhorlog}
\lim_{r\to 0^+} \frac{\rho(r)}{r \sqrt{-\log  r}} =1\,.
\end{equation}
The combination of \eqref{eq:mmLB} and \eqref{eq:ballInclusions} gives that there exists constants $c_1, c_2,c_3,c_4>0$ such that for all $x\in B^{\bar{g}}_{\varepsilon_0}(\mathsf{O})\setminus B^{\bar{g}}_{2R}(\mathsf{O})$ it holds
\begin{align}
\bar{\mm} (B^{\bar{g}}_R(x)) &\geqslant  c_1  R^2 \; \frac{\inf_{B^{\bar{g}}_R(x)} H} {\sup_{B^{\bar{g}}_R(x)} H} \geqslant  c_2  R^2 \; \frac{\inf_{B^{\bar{g}}_R(x)} -\log{r}} {\sup_{B^{\bar{g}}_R(x)} -\log r}, \nonumber \\
& \geqslant c_2 R^2  \; \frac{-\log{r(\rho(x)+R)}} { -\log r(\rho(x)-R)}, \nonumber \\
& \geqslant c_3 R^2  \; \frac{r(R)} { r(3R)},  \nonumber \\
&\geqslant c_4 R^2,   \label{mmBLB}
\end{align} 
where in the third line we used L'H\^opital's rule, and in the fourth line we used \eqref{eq:rhorlog}.  The claim \eqref{Claim1-p7} then follows by dominated convergence theorem by observing that the constant $c_4^{-1}$ is in $L^1(B^{\bar{g}}_{\varepsilon_0}(\mathsf{O}), \bar{\mm})$ and that \eqref{mmBLB} guarantees
$$
\frac{R^{2}}{\bar \mm(B_R(x))} \;  \chi_{B^{\bar{g}}_{\varepsilon_0}(\mathsf{O})\setminus B^{\bar{g}}_{2R}(\mathsf{O}) } (x) \leqslant c_4^{-1} \in L^1(B^{\bar{g}}_{\varepsilon_0}(\mathsf{O}), \bar{\mm}). 
$$

We now show the second claim \eqref{Claim2-p7}.
\\It is useful to perform a change of variables, so as to reduce the argument to a standard cone. Following the proof of \cite[Thm.\,3.2]{deluca-deponti-mondino-t}, we perform the change of variable \eqref{eq:rtorho}
 so that, in a neighbourhood of the  singular point $\mathsf{O}$, the Riemannian metric $\bar{g}$ becomes
\begin{equation}\label{eq: cone metricp7}
\bar{g}=\dd\rho^2+f(\rho)\dd s^2_{\mathbb{S}^1},
\end{equation}
and the measure takes the form
\begin{equation}\label{eq: cone weighted measp7}
\bar{\mm}=\sqrt{f(\rho)}\,\dd\mathsf{vol}_{\rho}\,\dd\mathsf{vol}_{\mathbb{S}^1},
\end{equation}
where $f(\rho)$ represents the factor $-\log(r)r^2$, expressed in the new variable $\rho$.  One can check that
\begin{equation}\label{eq:Propfrho}
0\leqslant f(\rho)\leqslant \rho^2, \quad \rho\mapsto f(\rho) \text{ is increasing for small } \rho.
\end{equation}
This estimate and the formula \eqref{eq: cone metricp7} of $\bar{g}$ yield that the distance $\di_{\bar{g}}$ associated to $\bar{g}$ satisfy the following properties:
\begin{itemize}
\item For every $\theta \in \mathbb{S}^1$,  it holds $\di_{\bar{g}}((\theta,\rho),\mathsf{O})=\rho$.
\item The distance $\di_{\bar{g}}$ is bounded above by the cone distance $\di_{C(\mathbb{S}^1)}$, i.e.~for every pair of points $x_0:=(\theta_0,\rho_0)$,  $x_1:=(\theta_1,\rho_1)$ it holds  
$$\di_{\bar{g}}(x_0,x_1)\leqslant \di_{C(\mathbb{S}^1)}(x_0,x_1):=\sqrt{\rho_0^2+\rho_1^2-2\rho_0\rho_1\cos(\di_{\mathbb{S}^1}(\theta_0,\theta_1)}).$$ 
\end{itemize} 
The second fact implies that, for every base point $x$ and every radius $R>0$,  we have the inclusion of metric balls
$
B^{C(\mathbb{S}^1)}_R(x)\subset B^{\bar{g}}_R(x),
$
which implies 
\begin{equation}\label{eq:volBallsp7}
\bar{\mm}(B^{C(\mathbb{S}^1)}_R(x)) \leqslant \bar{\mm}(B^{\bar{g}}_R(x)).
\end{equation}
Using \eqref{eq: cone weighted measp7} and \eqref{eq:Propfrho}, one can check that there exists $\varepsilon_0>0$ such that 
\begin{equation}\label{eq:mmBReps}
\bar{\mm}(B^{C(\mathbb{S}^1)}_R(x)) \geqslant \varepsilon_0  \, \bar{\mm}(B^{C(\mathbb{S}^1)}_R(\mathsf{O})), \quad \forall R\in (0, \varepsilon_0].
\end{equation}
A direct computation gives
\begin{equation}\label{eq:mmBO}
\bar{\mm}(B^{C(\mathbb{S}^1)}_R(\mathsf{O}))=2\pi\int_0^R \sqrt{f(\rho)}\, \dd \rho \geqslant \pi r(R)^2,
\end{equation}
where $r(R)$ is the value of the variable $r$ when $\rho=R$, with $r$ and $\rho$ are related by  \eqref{eq:rtorho}. 
Combining \eqref{eq:volBallsp7}, \eqref{eq:mmBReps} and \eqref{eq:mmBO}, we obtain
\begin{equation}\label{eq:Claimp7pf1}
\frac{R^2}{\bar \mm(B^{\bar g}_R(x))}  \leqslant  \frac{1}{\pi \varepsilon_0} \frac{R^2}{r(R)^2}\,,  \qquad \forall  R\in (0, \varepsilon_0], \ \ \forall x\in B^{\bar g}_{\varepsilon_0}(\mathsf{O}).
\end{equation}
Recalling that $(0,\varepsilon_0]\ni \rho\mapsto f(\rho)$ is increasing, from \eqref{eq:rhorlog} and \eqref{eq:Claimp7pf1} we have that
 \begin{align}
 \frac{R^2}{\bar \mm(B^{\bar g}_R(x))}  \bar \mm &\leqslant  \frac{1}{\pi \varepsilon_0} \frac{R^2}{r(R)^2} \sqrt{f(3R)} \, \dd\mathsf{vol}_{\rho}\,\dd\mathsf{vol}_{\mathbb{S}^1}, \nonumber \\
 & \leqslant   \frac{2}{\pi \varepsilon_0} \frac{r^2 |\log r|}{r^2} \sqrt{-\log(3r) (3r)^2} \sqrt{|\log r|} \, \dd\mathsf{vol}_{r}\,\dd\mathsf{vol}_{\mathbb{S}^1}\nonumber \\
& \leqslant \frac{6}{\pi \varepsilon_0} r\, |\log r| \sqrt{\log(3r) \log(r) }  \, \dd\mathsf{vol}_{r}\,\dd\mathsf{vol}_{\mathbb{S}^1} \nonumber \\
& \leqslant C  \, \dd\mathsf{vol}_{r}\,\dd\mathsf{vol}_{\mathbb{S}^1}  \qquad \forall  R\in (0, \varepsilon_0], \ \ \forall x\in B^{\bar g}_{2R}(\mathsf{O}),\label{eq:lastclaimp7}
 \end{align}
 for some constant $C=C(\varepsilon_0)$ independent of $R$, as long as $R\in (0, \varepsilon_0]$. Observing that \eqref{eq:rtorho} gives $\lim_{\rho\to 0^+} r(\rho)= 0$, it is clear that \eqref{eq:lastclaimp7} implies the second claim \eqref{Claim2-p7}.

\vspace{3mm}
\textbf{Case} $p=8$.
\\Let us start by recalling the expression of the metric and measure of a D$8$-brane. Given a positive constant  $h_{8}>0$, set
\begin{equation*}
H(r):= 1-h_{8} |r|, \quad \text{for } r\in [- (2h_{8})^{-1}, (2h_{8})^{-1}]. 
\end{equation*}
The metric and the measure have the following expressions
\begin{equation*}
\bar{g}=  H(r)\, \dd r^2, \quad  \bar{\mm}= H^{1/2} \dd \text{vol}_{\bar{g}}, \quad   \text{for } r\in [- (2h_{8})^{-1}, (2h_{8})^{-1}]. 
\end{equation*}
Thanks to the dimension reduction discussed before the proof of the case $p=7$,  \eqref{eq:claimR6mmBR} reduces to show that, for some $\varepsilon_0>0$,
\begin{equation}\label{eq:claimR1mmBR}
 \lim_{R\to 0^+} \int_{B^{\bar{g}}_{\varepsilon_0}(\mathsf{O})} \frac{R}{\bar \mm(B_R(x))}\dd \bar{\mm}=\int_{B^{\bar{g}}_{\varepsilon_0}(\mathsf{O})} \lim_{R\to 0^+} \frac{R}{\bar \mm(B_R(x))}\dd \bar\mm\,.
\end{equation}
It is useful to perform a change of variables, so to reduce the argument to the standard metric on an interval with a weighted measure. Following the proof of \cite[Thm.\,3.2]{deluca-deponti-mondino-t}, we perform the change of variable
\begin{equation*}
t(r):=- \frac{2}{3h_{8}} \sgn(r) \left[ (1-h_{8} |r|)^{3/2} -1 \right],
\end{equation*}
which transforms $\bar{g}$ into the euclidean metric on a segment in the real line and the measure  $\bar{\mm}$ into
\begin{equation}\label{eq:hatmm}
\hat{\mm}= \left( 1- \frac{3}{2} h_{8} |t| \right)^{1/3} \dd t.
 \end{equation}
 The claim \eqref{eq:claimR1mmBR} is thus equivalent to show that, for some $\varepsilon_0>0$,
\begin{equation}\label{eq:claimR1mmBR2}
 \lim_{R\to 0^+} \int_{[-\varepsilon_0, \varepsilon_0]} \frac{R}{\hat{\mm}([x-R, x+R])}\dd \hat{\mm}=\int_{[-\varepsilon_0, \varepsilon_0]}  \lim_{R\to 0^+}  \frac{R}{\hat{\mm}([x-R, x+R])}\dd \hat{\mm}.
\end{equation}
Notice that there exists $c_0=c_0(\varepsilon_0)>0$ independent of $x,R\in [-\varepsilon_0, \varepsilon_0]$ such that
\begin{equation}\label{eq:LBh8}
\inf_{t\in [x-R, x+R]} \left(1- \frac{3}{2} h_{8} |t|\right)^{1/3} \geqslant c_0,  \quad \text{ for all } x,R\in [-\varepsilon_0, \varepsilon_0],
\end{equation}
provided $\varepsilon_0>0$ is small enough, only depending on $h_8>0$.
The combination of \eqref{eq:hatmm} and \eqref{eq:LBh8} gives that
\begin{equation*}
\frac{R}{\hat{\mm}([x-R, x+R])} \leqslant \frac{1}{2c_0} ,  \quad \text{ for all } x,R\in [-\varepsilon_0, \varepsilon_0].
\end{equation*}
Since $\frac{1}{2c_0}\in L^1( [-\varepsilon_0, \varepsilon_0], \hat{\mm})$, the claim \eqref{eq:claimR1mmBR2} follows by dominated convergence theorem.

\end{proof}
% section weyl-rev (end)

\section{Conclusions} % (fold)
\label{sec:conc}
We have seen how to use consistency of gravity compactifications to derive the Weyl law for eigenvalues.
The argument proceeds by studying the gravitational potential between two point sources and requiring agreement at short distances among the four- and $D$-dimensional behavior. This results in a relation among the Green's function and the spectral data of the Laplacian (Sec.~\ref{sub:comp}) which we also derived in a mathematically precise way in Sec.~\ref{sub:grav-sc},  and repeat here:
\begin{equation}
	\ee^{f(y_0)}\lim_{r\to 0} r^{n}\sum_k  \psi^2_k(y_0)\ee^{- m_k r} = \frac{m_4^2}{m_D^{D-2}}  \frac{\Gamma(\frac{n+1}{2})}{\pi^{\frac{n+1}{2}}}.\label{eq:weylconcl}
\end{equation}
Equation \eqref{eq:weylconcl} implies the Weyl law once rendered independent from the eigenfunctions $\psi_k$, and in Sec.~\ref{sec:grav-weyl} we studied two ways to achieve this:  i) integrating it over the internal space and ii) studying the asymptotic behavior of $\psi_k(y_0)$ at large $k$.
However, a puzzle seems to arise from \eqref{eq:weylconcl} for warped compactifications: the 4-dimensional Planck mass $m_4$ depends on the weighted volume $V_f  = \int_X\sqrt{g_n} \ee^f$, while the Weyl law always depend on the geometric volume of the internal space. 

In the integrated argument i), this is easily resolved by the $\ee^{f(y_0)}$ factor on the left hand side of \eqref{eq:weylconcl}, which upon integration gives a weighted volume that compensates with the one inside $m_4^2$, leaving behind the geometric volume that appears in the Weyl law (Sec.~\ref{sub:ave}).

Method ii) instead is a local procedure, more aligned in spirit with the physical idea of studying the local behavior of the potential, and in Sec.~\ref{sub:wqe} it leads us to introduce the novel notion of \emph{Weighted Quantum Ergodicity} (WQE):
	\begin{equation}\label{eq:WQE-concl}
		\lim_{\substack{k\to \infty\\ k\notin e}}\frac{\int_B \sqrt{g}\ee^f\psi_k^2}{\int_{X} \sqrt{g}\ee^f\psi_k^2}\ = \  \frac{V(B)}{V(X)} \,,
	\end{equation}
for a negligible set $e\subset \mathbb{N}$. This extends the standard notion of Quantum Ergodicity, to which it reduces for $f = 0$, and implies that for large $k$ the $\psi_k^2(y_0)$ oscillates around $\ee^{-f(y_0)}$, as shown in Fig.~\ref{fig:WQE-1d} for an explicit example.

This latter property of $\psi_k^2(y_0)$ connects to the physics of gravity localization, as we started investigating in Sec.~\ref{sub:loc}. In this context, knowledge of the behavior of $\psi_k$ allows to infer whether gravity can be localized in four dimensions even in situations where the internal space is not of finite (warped) volume. However, very few results on the behavior of eigenfunctions are known, particularly in the physically important case of warped compactifications. Through the notion of WQE, the results in this paper provide novel insights about the asymptotic behavior of the $\psi_k$ for generic scenarios in which ergodicity holds. It would be interesting to expand this connection more in the future, for example analyzing in more detail the cases with many symmetries, in which, generically, ergodicity is not expected to hold.
\medskip

Finally, in Sec.~\ref{sec:weyl-rev} we rigorously proved the Weyl law for spaces with D$p$-brane singularities, for $p = 6,7,8$, which, as we formalized in Prop.~\ref{prop:Dbranediscrete}, are the cases in which the spectrum is discrete. 
For other branes, the spectrum is instead continuous, and we did not attempt a rigorous proof of a Weyl-type law for those cases. 
While for spaces with O-plane singularities the spectrum is always expected to be discrete, they are outside of the RCD class \cite[Sec.~6.4.1]{deluca-deponti-mondino-t-entropy} and thus outside of the class of spaces covered by current mathematical techniques. Nevertheless, the physical derivation still applies and we expect a similar validity of the Weyl law, albeit we cannot provide a rigorous proof yet.

% section conc (end)

\section*{Acknowledgements}

We thank Alex Belin, Naomi Gendler, Miguel Montero, Ignacio Ruiz, Irene Valenzuela, and Zhenbin Yang for discussions. GBDL is supported in part by the NSF Grant PHY-2310429. NDP is supported by the INdAM-GNAMPA project “Mancanza di regolarità e spazi non lisci:
studio di autofunzioni e autovalori”, CUP E53C23001670001. AM is supported by the ERC Starting Grant 802689 ``CURVATURE''. AT is supported in part by INFN and by MUR-PRIN contract 2022YZ5BA2.

\appendix

\section{Examples} % (fold)
\label{app:ex}

\subsection{Ergodic examples} % (fold)
\label{sub:ex-erg}

We will consider examples in one dimension, with $X=\mathbb{S}^1$. The metric can of course always to be taken to be the Euclidean distance, but we can still consider different choices of $f=-A$.

Consider first the piecewise-linear
\begin{equation}
	f= |x|
\end{equation}
on $X=([-1/2,1/2]/\!\!\sim) \cong \mathbb{S}^1$, with $\sim$ denoting identification of the two endpoints; $V(X)=1$. The eigenfunctions can be found explicitly: after imposing that they are in $C^1(\mathbb{S}^1)$, 
\begin{equation}\label{eq:ex-erg1}
	\psi_k= \ee^{-|x|/2} \left(c_+ \sin(2\pi k x)+ c_- (\sin(2\pi k x)+ 4 \pi c_- k \cos(2\pi k x))\right)\,,
\end{equation}
for $k\in \mathbb{N}$. The eigenvalues are $\lambda_k= \frac 14 + 4\pi^2 k^2$. Notice that this satisfies the Weyl law (\ref{eq:weyl-count}) after taking into account that each level has multiplicity two: $N(\lambda_k)\sim 2 k \sim \sqrt{\lambda_k}/\pi$ as $k\to\infty$. Indeed the leading term in $k$ is $4\pi^2 k^2$, the same as for the $f=0$ problem.

When considering (\ref{eq:WQE}) on an interval $B$, the $\ee^f=\ee^{|x|}$ in the integrals cancels with the $\ee^{-|x|/2}$ prefactor in (\ref{eq:ex-erg1}), and the integrand becomes a simple (translated) sine squared. This is equal to a constant plus an oscillating term; upon taking $k\to \infty$, the latter is suppressed. For example, taking $c_+=\sqrt2$, $c_-=0$ we have$\int_X \ee^f \psi_k^2=1$; if $B=[x_0- \delta x/2,x_0+ \delta x/2]\subset \{x>0\}$,
\begin{equation}
	\int_B \ee^f \psi_k^2 = \delta x - \frac1{2\pi k} \cos(4\pi k x_0)\sin(2\pi k \delta x)\,.
\end{equation}
The second term is a contribution from the oscillation, and it gets suppressed as $k\to \infty$, leaving a limit $\delta x = V(B)$, in agreement with (\ref{eq:WQE}).

With more complicated choices of $f$, analytic expressions for the $\psi_k^2$ are not always available (and even less for their integrals); however, a numerical study is still relatively straightforward. Plotting the wavefunctions at relatively large $k$, we see that they indeed oscillate around $(V_f/V)\ee^{-f}$. See Fig.~\ref{fig:WQE-1d} for one particular $f$. 

% subsection ex-erg (end)

\subsection{A non-ergodic example} % (fold)
\label{sub:ex-nerg}

As we mentioned, the WQE does not hold for a space that has many symmetries, just like the ordinary QE. Here we will present such an example, where we were able nevertheless to check that our gravitational formula (\ref{eq:grav-weyl}) holds anyway.

AdS$_7$ solutions \cite{afrt,10letter,cremonesi-t} have internal (Einstein frame) metric and weight function 
\begin{equation}\label{eq:ads7}
	\begin{split}
	 \dd s^2_\mathrm{E} &= \frac{(\dot \alpha^2- 2 \alpha \ddot \alpha)^{1/4}}{2^{1/8} \,9 \,\pi^{1/4}} \left(-\frac {\ddot \alpha}{\alpha}\right)^{7/8}\left( \dd z^2 + \frac{\alpha^2}{\dot \alpha^2- 2 \alpha \ddot \alpha}\dd s^2_{\mathbb{S}^2}\right)
		\, ,\\
		\ee^f &= \ee^{8A_\mathrm{E}}= \frac{2^{11}\sqrt2}{3^8 \pi} \sqrt{-\frac{ \alpha}{\ddot \alpha}} (\dot \alpha^2- 2 \alpha \ddot \alpha)\,.
	\end{split}
\end{equation}
$\alpha$ is a function on an interval with coordinate $z$; it is piecewise-cubic, with coefficients suitably quantized and obeying certain boundary conditions. 

Perhaps the simplest non-trivial example is $\alpha=\frac{27}2 n_0\pi^2 z (N^2 - z^2)$, with $z\in[0,N]$, where $n_0=2\pi F_0$ and $N$ are the Romans mass and NSNS three-form flux quanta respectively. There are $k=n_0 N$ D6-branes at $z=N$. The spin-two KK problem in this case can be solved analytically \cite[Sec.~4.2]{passias-t}.\footnote{The analysis was partially redone recently in a more modern language in \cite{lima-spin2}.} The eigenfunctions read
\begin{equation}\label{eq:psi-AdS7}
	\psi_{\ell,j}=\alpha^{\ell}  P^{(2\ell+1,\ell+1/2)}_j\left(-1+2z^2/N^2\right) Y_\ell^m \,,
\end{equation} 
where the $P$ are Jacobi polynomials, and the $Y$ are spherical harmonics on the $\mathbb{S}^2$ in (\ref{eq:ads7}). The eigenvalues read
\begin{equation}
	m^2_{\ell,j}= 4\ell(4\ell+6) + \frac83j(5+2j+6\ell)\,.
\end{equation}

The angular momentum associated to the symmetries of the $\mathbb{S}^2$ is conserved, and while the motion is chaotic, it cannot get all the way to the poles of the sphere.\footnote{In a geometry $-\dd t^2 + a^2 \dd z^2 + b^2 \dd s^2_{\mathbb{S}^2}$, the geodesic motion can be simplified using the conserved quantities $\dot t:=E$, $L^2= b^4 (\dot \phi^2 + \sin^2 \phi \dot\theta^2) $; this results in $a^2 z^2 + L^2 b^{-2}= E^2-1$, which has the form of a conserved energy for a particle moving in a one-dimensional potential. Compactness demands that $b$ goes to zero at the endpoints of the $z$ interval, so the effective potential $L^2 b^{-2}$ blows up at the endpoints, and the motion can never get there. The chaotic nature of the classical motion of a \emph{string} (rather than a particle) in this geometry was studied in \cite{filippas-nunez-vangorsel}.} Indeed the eigenfunctions (\ref{eq:psi-AdS7}) are of the WKB form, oscillating in a subinterval of  $[0,1]\ni x$, and decaying exponentially outside it; this reflects the classical analysis, just like in the $\mathbb{S}^2$ discussion in Sec.~\ref{sub:erg}. Thus the model does not display the WQE property.\footnote{As mentioned in Sec.~\ref{sub:erg}, it is sometimes possible to achieve a version of quantum ergodicity by considering random orthonormal bases; however in the present case such a basis would mix different eigenvalues, and would not be made of eigenfunctions, as in \cite{zelditch-QE-highdimension,maples}.}
Moreover, the spectrum depends on two integers $(\ell,j)$, and their ordering in terms of growing $m^2$ is not particularly elegant.

Nevertheless, given that both the $\psi_{\ell,j}$ and $m_{\ell,j}$ are known explicitly, we have managed to checked numerically that (\ref{eq:grav-weyl}) does hold even in this example. Summing $O(10^2)$ eigenvalues is enough to obtain a precision of $O(10^{-2})$.

% subsection ex-nerg (end)

% appendix ex (end)

\section{Heat equation and Weyl law} % (fold)
\label{app:heat}

The modern proofs of the Weyl's law are based on the heat equation. Here we try to give a rough idea of how this works in the unweighted case. 

In the spirit of the rough idea described in Sec.~\ref{sub:ave}, we write $\lambda_k \sim (\alpha k^\nu)^2$ for large $k$, and we try to determine the coefficients $\alpha$, $\nu$ by using the heat equation 
\begin{equation}
	\partial_t u(x,t) = -\Delta u(x,t)\,.
\end{equation}
A solution can be found by expanding $u= \sum u_k \psi_k$; the equation then fixes $u_k = u_k^0 \ee^{- \lambda_k t}$. 
A fundamental solution $u(x,y,t)$ is one for which $\lim_{t\to 0} u(x,y,t)= \delta(x-y)$. 
For this, we notice that $\delta(x-y)= \frac1{V(X)}\sum_k \psi_k(x)\psi_k(y)$, with the normalization $\int_{X} \sqrt{g} \psi_k \psi_l = \delta_{kl}$. This fixes the constants $u_k^0$. We infer that
\begin{equation}\label{eq:heat-psi2}
	u(x,y,t)= \frac1{V(X)}\sum_k \ee^{-\lambda_k t}\psi_k(x) \psi_k(y)\,.
\end{equation}
When $X=\mathbb{R}^n$, the spectrum is continuous; the analogue of the former expression is the Fourier transform: 
\begin{equation}\label{eq:fund-heat}
	u = \frac1{(2\pi)^{n/2}} \int \dd^n p \,\ee^{-\ii p \cdot x - ||p||^2 t} = 
	\frac1{(4\pi t)^{n/2}} \ee^{-||x-y||^2/4t}\,.
\end{equation}
On a general $X$, at small time the heat generated by the source at $t=0$ has not had time to explore all the space, so one expects it to behave as in (\ref{eq:fund-heat}). Integrating (\ref{eq:heat-psi2}), (\ref{eq:fund-heat}) over $X$ one then obtains as $t\to 0$
\begin{equation}\label{eq:heat0}
	\frac1{(4\pi t)^{n/2}} V(X)\quad \sim \quad\int_{X} \sqrt{g} u(x,x,t) = \sum_k \ee^{-\lambda_k t} \sim \sum_k \ee^{- \alpha^2 k^{2/\nu} t}\,.
\end{equation}
Since $t$ is small, we can approximate the sum by recalling (\ref{eq:sum-int}). Taking $\epsilon=(\alpha \sqrt{t})^{\nu}$, (\ref{eq:heat0}) becomes
\begin{equation}\label{eq:weyl-heat}
	\frac1{(4\pi t)^{n/2}} V(X) \sim \frac1{(\alpha \sqrt t)^{\nu}}\int_0^\infty \dd p \, \ee^{-p^{2/\nu}}
	= \frac{\Gamma\left(1+\nu/2\right)}{(\alpha \sqrt t)^{\nu}}\,.
\end{equation}
Comparing the powers of $\delta t$ yields $\nu=n$; the overall coefficient gives $\alpha=a$ as in (\ref{eq:weyl}).
% appendix heat (end)

\section{Proofs of Lemmas about limits of Green's function}\label{app:green}
In this Appendix, we prove the Lemmas \ref{lemma:ef0} and \ref{lemma:Rn} we used in Sec.~\ref{sub:grav-sc} to obtain equation \eqref{eq:scal-weyl}, which served as a basis for our argument to derive the Weyl law and for the introduction of the notion of weighted quantum ergodicity.

%\begin{lemma}\label{lemma:ef0App}
%	Call $G_{f, z_0}$ the Green's function of the operator $\Delta_f$ %centered at
%$z_0$, and $G_{0, z_0}$ the Green's function of the standard Laplacian
%$\Delta_0$ on the same $p$-dimensional Riemannian manifold $M_p$. If $f$ is smooth at $z = z_0$ , then
%\begin{equation}\label{eq:C1}
%  \lim_{z \rightarrow z_0} \frac{G_{f, z_0} (z)}{G_{0, z_0} (z)} = e^{- f
%  (z_0)} .
%\end{equation}
\begin{proof}[Proof of Lemma \ref{lemma:ef0}]
	By definition of Green's function, for any smooth test function $\xi$, we have
	\begin{equation}
		\xi(z_0) = \int_{M_p} \dd^p z \sqrt{g} \ee^f\, \xi \Delta_f\left(G_{f,z_0}\right)  = \int_{M_p} \dd^p z \sqrt{g} \ee^f\, G_{f,z_0} \Delta_f\left(\xi\right)  
	\end{equation}
	for any $z_0 \in M_p$.
Now, consider a ball $B(z_0; \varepsilon)$, centered at $z_0$ and with radius $\varepsilon$, and split the last integral on the right hand side in two pieces:
\begin{equation}\label{eq:lem1int}
	\int_{M_p} \dd^p z \sqrt{g} \ee^f\, G_{f,z_0} \Delta_f\left(\xi\right) = \int_{B(z_0; \varepsilon)} \dd^p z \sqrt{g} \ee^f\, G_{f,z_0} \Delta_f\left(\xi\right)+ \int_{M_p\setminus B(z_0; \varepsilon)} \dd^p z \sqrt{g} \ee^f\, G_{f,z_0} \Delta_f\left(\xi\right). 
\end{equation}

To prove \eqref{eq:C1}, we make a suitable choice  of test function $\xi$. Taking a $\xi^\prime$ with compact support away from $z_0$, we construct $\xi = (1+\xi^\prime)\zeta$, where $\zeta$ has compact support in $M_p$ and is identically equal to $1$ in a neighbourhood of $z = z_0$. For such a $\xi$, the first integral on the right hand side of \eqref{eq:lem1int} is regular, and thus
\begin{equation}
	\lim_{\varepsilon \to 0}\int_{B(z_0; \varepsilon)} \dd^p z \sqrt{g} \ee^f\, G_{f,z_0} \Delta_f\left(\xi\right) = 0 \;.
\end{equation}
We now analyze the second integral. Since we are now away from the singularity of the Green's function, we can freely integrate by parts twice, obtaining 
\begin{equation}
	\int_{M_p\setminus B(z_0; \varepsilon)} \dd^p z \sqrt{g} \ee^f\, G_{f,z_0} \Delta_f\left(\xi\right)  = \int_{\mathbb{S}^{p-1}(z_0,\varepsilon)}\sqrt{g} \ee^f \left(\xi \eta \cdot \nabla G_{f, z_0}-  G_{f, z_0}\eta \cdot \nabla\xi  \right)
\end{equation}
where $\eta$ is the vector normal to the boundary and the integral is on the sphere with radius $\varepsilon$ centered at $z = z_0$. We also used the fact that the only boundary terms are at $|z-z_0| =  \varepsilon$, since other possible terms (e.g.~at infinity) would vanish by virtue of $\xi$ of being of compact support. In addition, in the limit $\varepsilon \to 0$ the second term in the parenthesis also vanishes, since for the above construction $\nabla \xi$ vanishes on a neighborhood of $z_0$.
Putting everything together, we thus have
\begin{equation}\label{eq:resGflemma}
	\xi(z_0) = \xi(z_0) \ee^{f(z_0)}\lim_{\varepsilon \to 0 }  \int_{\mathbb{S}^{p-1}(z_0,\varepsilon)}\sqrt{g}  \eta \cdot \nabla G_{f, z_0}\,.
\end{equation}

We can now repeat the same argument from scratch, but for $\Delta_0$ and its corresponding Green's function $G_{0,z_0}$. Doing so, results in \eqref{eq:resGflemma} with $f = 0$. 
Combining the two results, we find 
\begin{equation}
	e^{-f(z_0)} = \frac{\lim_{\varepsilon \to 0 }  \int_{\mathbb{S}^{p-1}(z_0,\varepsilon)}\sqrt{g}  \eta \cdot \nabla G_{f, z_0}}{\lim_{\varepsilon \to 0 }  \int_{\mathbb{S}^{p-1}(z_0,\varepsilon)}\sqrt{g}  \eta \cdot \nabla G_{0, z_0}} = \lim_{\varepsilon \to 0 }\frac{  \int_{\mathbb{S}^{p-1}(z_0,\varepsilon)}\sqrt{g}  \eta \cdot \nabla G_{f, z_0}}{  \int_{\mathbb{S}^{p-1}(z_0,\varepsilon)}\sqrt{g}  \eta \cdot \nabla G_{0, z_0}}\;.
\end{equation}

The result then follows by going in Riemann normal coordinates around $z_0$ and applying L'H\^{o}pital's rule.

\end{proof}
%\end{lemma}

With a similar technique, we can also prove Lemma \ref{lemma:Rn}. %, which we repeat here. 
% \begin{lemma}
%	Call $G_{0, z_0}$  the Green's function of the Laplacian on a smooth $p$-dimensional Riemannian manifold $M_p$, centered at $z = z_0$, then 
%	\begin{equation} \label{eq:lemmaRn}
%		 \lim_{z \to z_0}|z-z_0|^{p-2}G_{0, z_0}(z)  = - \frac{1}{4} \pi^{- %\frac{p}{2}} \Gamma \left( \frac{p}{2} - 1 \right)\;
%	\end{equation} 
%	Equivalently, $G_{0, z_0}$ approaches the Green's function in $\mathbb{R}^p$, for $z\to z_0$.  
\begin{proof}[Proof of Lemma \ref{lemma:Rn}]
	Following the proof of Lemma \ref{lemma:ef0} and in particular using Eq.~\eqref{eq:resGflemma} with $f = 0$, we can write 
\begin{equation}\label{eq:resLemma2}
	1 =  \lim_{\varepsilon \to 0 }  \int_{\mathbb{S}^{p-1}(z_0,\varepsilon)}\sqrt{g}  \eta \cdot \nabla G_{0, z_0}\;.
\end{equation}
Since we are taking the limit $\varepsilon \to 0$, we can use Riemann normal coordinates centered around $z = z_0$. This shows that $G_{0,z_0}$ has to approach the Green's function of the Laplacian in $\mathbb{R}^p$ for \eqref{eq:resLemma2} to be valid in this limit. Since the power behavior on the left hand side, and the constants on the right hand side in \eqref{eq:lemmaRn} agree with the ones for the Green's function of the Laplacian in $\mathbb{R}^p$,  this concludes the proof.

We can be more explicit and also check this final statement directly from Eq.~\eqref{eq:resLemma2}. 
In Riemann normal coordinates we have $\sqrt{g} = \text{vol}(\mathbb{S}^{p-1})r^{p-1}$, $\eta \cdot \nabla G_{0, z_0} = \partial_r G_{0, z_0}(r)$, where $r = |z-z_0|.$
Plugging in \eqref{eq:resLemma2} and performing the integral over the sphere gives
\begin{equation}\label{eq:finlemma2}
	1 = \text{Vol}(\mathbb{S}^{p-1}) \lim_{r \to 0}\frac{\partial_r G_{0, z_0}(r)}{r^{1-p}}= \text{Vol}(\mathbb{S}^{p-1})(2-p) \lim_{r \to 0}\frac{G_{0, z_0}(r)}{r^{2-p}}\;,
\end{equation}
where in the last step we used L'H\^{o}pital's rule. Expanding the constants in \eqref{eq:finlemma2} then proves the result.

\end{proof}
%\end{lemma}

\bibliography{at}

\providecommand{\href}[2]{#2}\begin{thebibliography}{10}

\bibitem{schnirelman}
A.~Schnirelman, ``Ergodic properties of eigenfunctions,'' {\em Usp. Mat. Nauk}
  {\bf 29} (1974) 181--182.

\bibitem{colindeverdiere}
Y.~Colin~de Verdi{\`e}re, ``Ergodicit{\'e} et fonctions propres du
  {Laplacien},'' {\em Communications in Mathematical Physics} {\bf 102} (1985)
  497--502.

\bibitem{zelditch-surfaces}
S.~Zelditch, ``Uniform distribution of eigenfunctions on compact hyperbolic
  surfaces,'' {\em Duke Mathematical Journal} {\bf 55} (Jan., 1987) 919--941.

\bibitem{bachas-estes}
C.~Bachas and J.~Estes, ``{Spin-2 spectrum of defect theories},'' {\em JHEP}
  {\bf 06} (2011) 005,
\href{http://arXiv.org/abs/1103.2800}{{\tt 1103.2800}}.
%%CITATION = ARXIV:1103.2800;%%.

\bibitem{csaki-erlich-hollowood-shirman}
C.~Csaki, J.~Erlich, T.~J. Hollowood, and Y.~Shirman, ``{Universal aspects of
  gravity localized on thick branes},'' {\em Nucl. Phys.} {\bf B581} (2000)
  309--338,
\href{http://arXiv.org/abs/hep-th/0001033}{{\tt hep-th/0001033}}.
%%CITATION = HEP-TH/0001033;%%.

\bibitem{ivrii-review}
V.~Ivrii, ``100 years of {Weyl}'s law,'' {\em Bulletin of Mathematical
  Sciences} {\bf 6} (Aug., 2016) 379--452,
  \href{http://arXiv.org/abs/1608.03963}{{\tt 1608.03963}}.

\bibitem{hinterbichler-levin-zukowski}
K.~Hinterbichler, J.~Levin, and C.~Zukowski, ``{Kaluza--Klein} towers on
  general manifolds,'' {\em Phys. Rev.} {\bf D89} (2014), no.~8, 086007,
\href{http://arXiv.org/abs/1310.6353}{{\tt 1310.6353}}.
%%CITATION = ARXIV:1310.6353;%%.

\bibitem{duff-nilsson-pope}
M.~J. Duff, B.~E.~W. Nilsson, and C.~N. Pope, ``{Kaluza--Klein} supergravity,''
  {\em Phys. Rept.} {\bf 130} (1986)
1--142.
%%CITATION = PRPLC,130,1;%%.

\bibitem{deluca-deponti-mondino-t-entropy}
G.~B. De~Luca, N.~De~Ponti, A.~Mondino, and A.~Tomasiello, ``{Gravity from
  thermodynamics: optimal transport and negative effective dimensions},'' {\em
  SciPost Phys.} {\bf 15} (2023) 039,
  \href{http://arXiv.org/abs/2212.02511}{{\tt 2212.02511}}.

\bibitem{feller-vol2}
W.~Feller, {\em An introduction to probability theory and its applications,
  Volume 2}, vol.~81.
\newblock John Wiley \& Sons, 1991.

\bibitem{tao-scars}
T.~Tao, ``Hassell's proof of scarring for the {Bunimovich} stadium,'' 2008.
\newblock
  \url{https://terrytao.wordpress.com/2008/07/07/hassells-proof-of-scarring-for-the-bunimovich-stadium/#more-426}.

\bibitem{reichl}
L.~E. Reichl, {\em The Transition to Chaos}.
\newblock Springer, 2021.

\bibitem{zelditsch-S2}
S.~Zelditch, ``{Quantum ergodicity on the sphere},'' {\em Communications in
  Mathematical Physics} {\bf 146} (1992), no.~1, 61 -- 71.

\bibitem{zelditch-QE-highdimension}
S.~Zelditch, ``Quantum ergodicity of random orthonormal bases of spaces of high
  dimension,'' {\em Philosophical Transactions of the Royal Society A:
  Mathematical, Physical and Engineering Sciences} {\bf 372} (2014), no.~2007,
  20120511.

\bibitem{maples}
K.~Maples, ``Quantum unique ergodicity for random bases of spectral
  projections,'' {\em Math. Res. Lett.} (2013), no.~6, 1115--1124,
  \href{http://arXiv.org/abs/1306.3329}{{\tt 1306.3329}}.

\bibitem{deluca-deponti-mondino-t}
G.~B. De~Luca, N.~De~Ponti, A.~Mondino, and A.~Tomasiello, ``{Cheeger bounds on
  spin-two fields},'' {\em JHEP} {\bf 12} (2021) 217,
  \href{http://arXiv.org/abs/2109.11560}{{\tt 2109.11560}}.

\bibitem{randall-sundrum2}
L.~Randall and R.~Sundrum, ``{An alternative to compactification},'' {\em Phys.
  Rev. Lett.} {\bf 83} (1999) 4690--4693,
\href{http://arXiv.org/abs/hep-th/9906064}{{\tt hep-th/9906064}}.
%%CITATION = HEP-TH/9906064;%%.

\bibitem{karch-randall}
A.~Karch and L.~Randall, ``{Locally localized gravity},'' {\em JHEP} {\bf 05}
  (2001) 008,
\href{http://arXiv.org/abs/hep-th/0011156}{{\tt hep-th/0011156}}.
%%CITATION = HEP-TH/0011156;%%.

\bibitem{verlinde-RS2}
H.~L. Verlinde, ``{Holography and compactification},'' {\em Nucl. Phys. B} {\bf
  580} (2000) 264--274, \href{http://arXiv.org/abs/hep-th/9906182}{{\tt
  hep-th/9906182}}.

\bibitem{chan-paul-verlinde}
C.~S. Chan, P.~L. Paul, and H.~L. Verlinde, ``{A Note on warped string
  compactification},'' {\em Nucl. Phys. B} {\bf 581} (2000) 156--164,
  \href{http://arXiv.org/abs/hep-th/0003236}{{\tt hep-th/0003236}}.

\bibitem{bachas-lavdas2}
C.~Bachas and I.~Lavdas, ``Massive anti-de {Sitter} gravity from string
  theory,'' {\em JHEP} {\bf 11} (2018) 003,
\href{http://arXiv.org/abs/1807.00591}{{\tt 1807.00591}}.
%%CITATION = ARXIV:1807.00591;%%.

\bibitem{deluca-deponti-mondino-t-entropyproof}
G.~B. De~Luca, N.~De~Ponti, A.~Mondino, and A.~Tomasiello, ``{To appear},''.

\bibitem{deluca-deponti-mondino-t-L2}
G.~B. De~Luca, N.~De~Ponti, A.~Mondino, and A.~Tomasiello, ``{Harmonic
  functions and gravity localization},'' {\em JHEP} {\bf 09} (2023) 127,
  \href{http://arXiv.org/abs/2306.05456}{{\tt 2306.05456}}.

\bibitem{Amb}
L.~Ambrosio, ``Calculus, heat flow and curvature-dimension bounds in metric
  measure spaces,'' in {\em Proceedings of the {I}nternational {C}ongress of
  {M}athematicians---{R}io de {J}aneiro 2018. {V}ol. {I}. {P}lenary lectures},
  pp.~301--340.
\newblock World Sci. Publ., Hackensack, NJ, 2018.

\bibitem{mondino-naber}
A.~Mondino and A.~Naber, ``Structure theory of metric measure spaces with lower
  {R}icci curvature bounds,'' {\em J. Eur. Math. Soc. (JEMS)} {\bf 21} (2019),
  no.~6, 1809--1854, \href{http://arXiv.org/abs/1405.2222}{{\tt 1405.2222}}.

\bibitem{brue-semola}
E.~Bru\'{e} and D.~Semola, ``Constancy of the dimension for {${\rm RCD}(K,N)$}
  spaces via regularity of {L}agrangian flows,'' {\em Comm. Pure Appl. Math.}
  {\bf 73} (2020), no.~6, 1141--1204,
  \href{http://arXiv.org/abs/1804.07128}{{\tt 1804.07128}}.

\bibitem{ambrosio-honda-tewodrose}
L.~Ambrosio, S.~Honda, and D.~Tewodrose, ``Short-time behavior of the heat
  kernel and {Weyl}'s law on $\mathrm{RCD}^*(k, n)$-spaces,'' {\em Annals of
  Global Analysis and Geometry} {\bf 53} (2018) 97--119,
  \href{http://arXiv.org/abs/1701.03906}{{\tt 1701.03906}}.

\bibitem{zhang-zhu}
H.-C. Zhang and X.-P. Zhu, ``Weyl's law on {$\mathsf{RCD}(K,N)$} metric measure
  spaces,'' {\em Comm. Anal. Geom.} {\bf 27} (2019), no.~8, 1869--1914,
  \href{http://arXiv.org/abs/1701.01967}{{\tt 1701.01967}}.

\bibitem{IwKiYo23}
A.~Iwahashi, Y.~Kitabeppu, and A.~Yonekura, ``One dimensional {$\mathsf{RCD}$}
  spaces always satisfy the regular {W}eyl's law,'' {\em Proc. Amer. Math.
  Soc.} {\bf 151} (2023), no.~11, 4923--4934,
  \href{http://arXiv.org/abs/2302.09494}{{\tt 2302.09494}}.

\bibitem{DHPW23}
X.~Dai, S.~Honda, J.~Pan, and G.~Wei, ``Singular {W}eyl's law with {R}icci
  curvature bounded below,'' {\em Trans. Amer. Math. Soc. Ser. B} {\bf 10}
  (2023) 1212--1253, \href{http://arXiv.org/abs/2208.13962}{{\tt 2208.13962}}.

\bibitem{ChPrRi24}
Y.~Chitour, D.~Prandi, and L.~Rizzi, ``Weyl's law for singular {R}iemannian
  manifolds,'' {\em J. Math. Pures Appl. (9)} {\bf 181} (2024) 113--151,
  \href{http://arXiv.org/abs/1903.05639}{{\tt 1903.05639}}.

\bibitem{pesenson20}
I.~Z. Pesenson, ``A weak {W}eyl's law on compact metric measure spaces,'' {\em
  J. Pseudo-Differ. Oper. Appl.} {\bf 11} (2020), no.~4, 1447--1463,
  \href{http://arXiv.org/abs/1912.11093}{{\tt 1912.11093}}.

\bibitem{GMS}
N.~Gigli, A.~Mondino, and G.~Savar\'{e}, ``Convergence of pointed non-compact
  metric measure spaces and stability of {R}icci curvature bounds and heat
  flows,'' {\em Proc. Lond. Math. Soc. (3)} {\bf 111} (2015), no.~5,
  1071--1129, \href{http://arXiv.org/abs/1311.4907}{{\tt 1311.4907}}.

\bibitem{deponti-mondino}
N.~De~Ponti and A.~Mondino, ``Sharp {C}heeger-{B}user type inequalities in
  {$\mathsf{RCD}(K,\infty)$} spaces,'' {\em J. Geom. Anal.} {\bf 31} (2021),
  no.~3, 2416--2438, \href{http://arXiv.org/abs/1902.03835}{{\tt 1902.03835}}.

\bibitem{cianchi-mazya}
A.~Cianchi and V.~Maz'ya, ``On the discreteness of the spectrum of the
  {Laplacian} on noncompact {Riemannian} manifolds,'' {\em Journal of
  differential geometry} {\bf 87} (2011), no.~3, 469--492.

\bibitem{passias-richmond}
A.~Passias and P.~Richmond, ``{Perturbing AdS$_6 \times_w S^4$: linearised
  equations and spin-2 spectrum},'' {\em JHEP} {\bf 07} (2018) 058,
\href{http://arXiv.org/abs/1804.09728}{{\tt 1804.09728}}.
%%CITATION = ARXIV:1804.09728;%%.

\bibitem{afrt}
F.~Apruzzi, M.~Fazzi, D.~Rosa, and A.~Tomasiello, ``{All AdS$_7$ solutions of
  type II supergravity},'' {\em JHEP} {\bf 1404} (2014) 064,
\href{http://arXiv.org/abs/1309.2949}{{\tt 1309.2949}}.
%%CITATION = ARXIV:1309.2949;%%.

\bibitem{10letter}
F.~Apruzzi, M.~Fazzi, A.~Passias, A.~Rota, and A.~Tomasiello, ``Six-dimensional
  superconformal theories and their compactifications from type {IIA}
  supergravity,'' {\em Phys. Rev. Lett.} {\bf 115} (2015), no.~6, 061601,
\href{http://arXiv.org/abs/1502.06616}{{\tt 1502.06616}}.
%%CITATION = ARXIV:1502.06616;%%.

\bibitem{cremonesi-t}
S.~Cremonesi and A.~Tomasiello, ``{6d holographic anomaly match as a continuum
  limit},'' {\em JHEP} {\bf 05} (2016) 031,
\href{http://arXiv.org/abs/1512.02225}{{\tt 1512.02225}}.
%%CITATION = ARXIV:1512.02225;%%.

\bibitem{passias-t}
A.~Passias and A.~Tomasiello, ``{Spin-2 spectrum of six-dimensional field
  theories},'' {\em JHEP} {\bf 12} (2016) 050,
\href{http://arXiv.org/abs/1604.04286}{{\tt 1604.04286}}.
%%CITATION = ARXIV:1604.04286;%%.

\bibitem{lima-spin2}
M.~Lima, ``Spin-2 universal minimal solutions on type {IIA} and {IIB}
  supergravity,'' \href{http://arXiv.org/abs/2310.16536}{{\tt 2310.16536}}.

\bibitem{filippas-nunez-vangorsel}
K.~Filippas, C.~N\'u\~nez, and J.~Van~Gorsel, ``{Integrability and holographic
  aspects of six-dimensional $\mathcal{N}=(1,0)$ superconformal field
  theories},'' {\em JHEP} {\bf 06} (2019) 069,
  \href{http://arXiv.org/abs/1901.08598}{{\tt 1901.08598}}.

\end{thebibliography}
\bibliographystyle{at}

\end{document}